\documentclass[letterpaper,twocolumn,10pt]{article}
\usepackage{usenix}
\usepackage{fancyhdr}
\usepackage{epsfig,endnotes}
\usepackage{amsmath,amssymb,amsthm}
\usepackage{multirow}
\usepackage{balance}
\usepackage[english]{babel}
\usepackage{algorithm,algpseudocode}
\usepackage{graphicx}
\usepackage{tablefootnote}
\usepackage{tikz}
\usepackage{url}
\usepackage{color}
\usepackage[export]{adjustbox}
\usepackage{mathptmx}
\usepackage{times}
\usepackage{hyperref}
\usepackage[hyphenbreaks]{breakurl}
\usepackage{subcaption}
\usepackage{balance}
\usepackage{booktabs}
\usepackage{longtable}
\usepackage{caption}
\usepackage{subcaption}
\usepackage{bm}

\DeclareMathOperator*{\E}{\mathbb{E}}

\newtheorem{theorem}{Theorem}

\newcommand{\system}{{\em Dolos}}

\newcommand{\para}[1]{{\vspace{4pt} \bf \noindent #1 \hspace{6pt}}}

\newcommand{\abedit}[1]{{\color{black} #1}}
\newcommand{\shawn}[1]{{\color{black} #1}}
\newcommand{\htedit}[1]{{\color{black} #1}}
\newcommand{\bzedit}[1]{{\color{black} #1}}

\newenvironment{packed_itemize}{
\begin{list}{\labelitemi}{\leftmargin=1em}
  \setlength{\itemsep}{2pt}
  \setlength{\parskip}{0pt}
  \setlength{\parsep}{0pt}
  \setlength{\headsep}{0pt}
  \setlength{\topskip}{0pt}
  \setlength{\topmargin}{0pt}
  \setlength{\topsep}{0pt}
  \setlength{\partopsep}{0pt}
}{\end{list}}

\newcommand{\eg}{{e.g.,\ }}
\newcommand{\ie}{{i.e., }}
\newcommand{\etal}{{\em et al.\ }}
\newcommand{\datasetsmall}{{\tt Sirinam}}
\newcommand{\datasetlarge}{{\tt Rimmer}}

\newcommand{\dflarge}{{\tt DF$_{\text{\tt Rimmer}}$}}
\newcommand{\dfsmall}{{\tt DF$_{\text{\tt Sirinam}}$}}
\newcommand{\varlarge}{{\tt VarCNN$_{\text{\tt Rimmer}}$}}
\newcommand{\varsmall}{{\tt VarCNN$_{\text{\tt Sirinam}}$}}

\newcommand{\secspace}{\vspace{0in}}

\newfont{\mycrnotice}{ptmr8t at 7pt}
\newfont{\myconfname}{ptmri8t at 7pt}

\begin{document}

\title{A Real-time Defense against Website Fingerprinting Attacks}
\author{
{\rm Shawn Shan}\\
University of Chicago \\
\and
{\rm Arjun Nitin Bhagoji}\\
University of Chicago
\and
{\rm Haitao Zheng}\\
University of Chicago
\and
{\rm Ben Y.\ Zhao}\\
University of Chicago
}

\maketitle

\thispagestyle{plain}

\begin{abstract}
  Anonymity systems like Tor are vulnerable to Website Fingerprinting (WF)
  attacks, where a local passive eavesdropper infers the victim's
  activity. Current WF attacks based on deep learning classifiers have
  successfully overcome numerous proposed defenses. While recent defenses
  leveraging adversarial examples offer promise, \htedit{these adversarial examples} can only be
  computed 
  after the network session has concluded, thus offer users little protection
  in practical settings.

  We propose \system{}, a system that modifies user
  network traffic in real time to successfully evade WF attacks.
  \htedit{\system{} injects dummy packets into  traffic traces} by computing
  \textit{input-agnostic adversarial patches} that disrupt deep learning
  classifiers used in WF attacks. \htedit{Patches} are then applied to alter
  and protect user traffic in real time. Importantly, these \htedit{patches} are
  parameterized by a user-side secret, ensuring that attackers
  cannot use adversarial training to defeat \system{}. We experimentally
  demonstrate that \system{} provides 94+\% protection against
  state-of-the-art WF attacks under a variety of settings.  Against prior defenses, \system{} outperforms in terms of higher
  protection performance and lower information leakage and bandwidth overhead.
  Finally, we show that \system{} is robust against a variety of adaptive
  countermeasures to detect or disrupt the defense.
\end{abstract}

\vspace{-0.1in}
\section{Introduction}
\label{sec:intro}
\vspace{-0.1in}

Website-fingerprinting (WF) attacks are traffic analysis attacks that allow
eavesdroppers to identify websites visited by a user, despite their use of
privacy tools such as VPNs or the Tor anonymity
system~\cite{wang2014effective, hayes2016k}. The attacker identifies webpages
in an encrypted connection by analyzing and recognizing network traffic
patterns. These attacks have grown more powerful over time, improving in
accuracy and scale. The most recent variants can overwhelm existing defenses,
by training deep neural network (DNN) classifiers to identify the destination website
given a network trace.  In real world settings, WF attacks have proven
effective in identifying traces in the wild from a large number of candidate
websites using limited data~\cite{bhat2019var,sirinam2019triplet}.

There is a long list of defenses that have been proposed and then later
defeated by DNN based WF attacks. First, a class
of defenses obfuscate traces by introducing randomness~\cite{gong2020zero,
  juarez2016toward,dyer2012peek,wang2017walkie}. These obfuscation\abedit{-based} defenses
have been proven ineffective (< 60\% \abedit{protection}) against DNN based
attacks~\cite{bhat2019var,sirinam2018deep}. Other defenses have proposed
randomizing HTTP requests~\cite{cherubin2017bayes} or scattering traffic
redirection across different Tor
nodes~\cite{de2020trafficsliver,henri2020protecting}. Again, these defenses
provide poor protection against DNN-based attacks (<50\%
protection). Unsurprisingly, the success of DNN attacks have derailed efforts
to deploy WF defenses on Tor (\eg Tor stopped the implementation of WTF-PAD
after it was broken by the DF
attack~\cite{sirinam2018deep,wtfpad2015takedown})

Against these strong DNN attacks, the only defenses to show promise are recent
proposals that apply adversarial examples to \htedit{mislead} WF classification
models~\cite{imani2019mockingbird,hou2020wf}. Adversarial examples are known
weaknesses of DNNs, where a small, carefully tuned change to an input can
dramatically alter the DNN's output. These vulnerabilities have been
studied intensely in the adversarial ML community, and now are generally regarded
a fundamental property of DNNs~\cite{bugs-features,shafahi2018adversarial}. 

Unfortunately, defenses based on adversarial examples have one glaring
limitation. To be effective, adversarial examples must be crafted
individually for each input~\cite{moosavi2017universal}. In the WF context, the ``input'' is the entire
network traffic trace. Thus a defense built using adversarial examples
requires the entire traffic trace to be completed, before it can compute the
precise perturbation necessary to mislead the attacker's WF classifier. This
is problematic, since real world attackers will observe user traffic in
real time, and be unaffected by a defense that \htedit{can only take action
after the fact. }



In this paper,  we propose {\system{}}, a practical and effective defense
against WF attacks that can be applied to network traffic in real-time. The
key insight in \system{} is the application of the concept of
\textit{trace-agnostic patches} to WF defenses, derived from the concept of
adversarial patches. While adversarial examples are customized for each
input, adversarial patches can cause misclassifications when applied to a
wide range of input values.  In the context of a WF defense, ``patches'' are
pre-computed sequences of \htedit{dummy} packets that protect all network traces of visits
to a specific website. A patch is applied to an active, ongoing network
connection, i.e. in ``real time.'' As we will show, \system{} generates
patches parameterized by \shawn{a user-side secret such as private keys}. Unlike traditional
patches or universal perturbations, patches parameterized by a user's
secret cannot be
overcome by attackers, unless
they compromise the user's secret.




Our work describes experiences in designing and
evaluating \system{}, and makes four key contributions:
\begin{packed_itemize} \vspace{-0.06in}
\item We propose \system{}, a new WF defense that is highly effective against
  the strongest attacks using deep learning classifiers. More importantly,
  \system{} precomputes patches before the network connection, and applies
  the patch to protect the user in real time.  
\item We introduce a \shawn{secret-based} parameterization mechanism for adversarial
  patches, which allows a user to generate randomized patches that cannot be
  \htedit{correctly} guessed without knowledge of the \shawn{user-side secret}. We show that this gives \system{}
  strong resistance to \htedit{WF attacks that apply adversarial training
    with known adversarial 
    patches and perturbations},  which can otherwise defeat
  patches and universal perturbations~\cite{nasr2020blind}.
\item We evaluate \system{} under a variety of settings. Regardless of how
  the attacker trains its classifier (handcrafted features, deep learning, or
  adversarial training), \system{} provides 94+\% protection against
  the attack.  Compared to three state-of-the-art 
  defenses~\cite{gong2020zero,imani2019mockingbird,juarez2016toward}, 
  \system{} significantly outperforms all three in key metrics: overhead, 
  protection performance, and information leakage. 
  \item Finally, we consider attackers with full knowledge of our 
  defense (\ie full access to source code but not \shawn{the user secret}), and demonstrate that \system{} is
  robust against multiple adaptive attacks and countermeasures. 
  \vspace{-0.06in}
\end{packed_itemize}


\section{Web Fingerprinting Attacks and Defenses}
\label{sec:back}


\begin{figure}[t]
  \centering
  \includegraphics[width=1.0\linewidth]{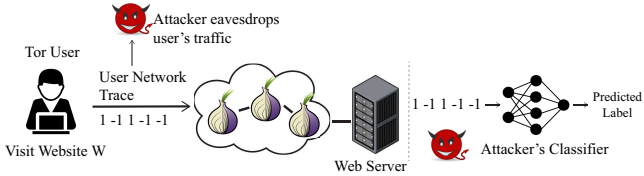}
  \caption{A WF attacker, positioned between user and the Tor network,
    eavesdrops on user network traffic. After the connection terminates, the
    attacker classifies the entire network trace using a pretrained WF attack
    classifier.}
  \label{fig:system_overview}
\end{figure}

\subsection{Website Fingerprinting Attacks}
\label{sec:attack}

In WF attacks, an attacker tries to identify the destination of user traffic
routed through an encrypted connection to Tor or a VPN. The attacker
passively eavesdrops on the user connection, and once the session is complete,
feeds the network trace as input to a \htedit{machine learning} classifier to identify the
website visited. Even though packets across the connection are encrypted and padded
to the same size, attackers can distinguish traces as sequences of packets
and their direction. The attacker's \htedit{machine learning} classifier is trained on
packet traces generated by visiting a large, pre-determined set of websites
beforehand. 

\para{\htedit{Attacks via Hand-crafted Features.}} Panchenko 
\etal~\cite{panchenko2011website} 
proposed the first effective WF attack against 
Tor traffic using a support vector machine with hand-crafted 
features. Follow-up work proposed stronger 
attacks~\cite{cai2012touching,hayes2016k,
herrmann2009website,liberatore2006inferring,lu2010website,
panchenko2016website,sun2002statistical,wang2014effective} by improving 
the feature set and using different 
classifier architectures. The most effective 
attacks based on hand crafted features such as
k-NN~\cite{wang2014effective}, CUMUL~\cite{panchenko2016website}
, and k-FP~\cite{hayes2016k} achieve
over $90\%$ accuracy in identifying websites 
based on network traces.

\para{\htedit{Attacks via DNN Classifiers.} }
Recent 
work~\cite{sirinam2018deep,bhat2019var,abe2016fingerprinting,Rimmer2018}
leverages 
deep neural networks (DNNs) to perform more powerful WF attacks. 
DNNs automatically extract features 
from raw network traces, and outperform 
previous WF attacks based on hand-crafted features. 
Two of the most successful DNN-based attacks are 
\textit{Deep Fingerprinting (DF)}~\cite{sirinam2018deep} 
and \textit{VarCNN}~\cite{bhat2019var}. \textit{DF} 
leverages a deep convolutional neural 
network for classification, and reaches 
over $98\%$ accuracy on undefended traces, or traces 
defended using existing defenses 
(\S\ref{sec:defenses}). \textit{Var-CNN}~\cite{bhat2019var} 
further improves the attack performance by using a
large residual network architecture, and can 
achieve high performance even with limited training 
data. Later in \S\ref{sec:eval}, our experiments show
that \system{} is effective at defending against both 
of these state-of-the-art attacks, as well as against 
those using hand-crafted features.

\subsection{Website Fingerprinting Defenses}
\label{sec:defenses}

WF defenses modify (add, delay, or reroute)
packets to prevent the identification of the 
destination website using  trace analysis.
Broadly speaking, defenses either obfuscate traces 
based on expert-designed noise heuristics or leverage adversarial 
perturbations to evade \htedit{machine learning} classifiers. 

\para{Defenses via Trace Obfuscation.} A number of defenses aim to obfuscate
traces to increase the difficulty of classification.  This obfuscation is
performed either at the application layer or the network layer. Since these
defenses do not specifically target the attacker's classifier, they generally
afford low protection (<60\%) against state-of-the-art WF attacks.

In \textit{application layer obfuscation}, the defender introduces randomness
into HTTP requests or the Tor routing algorithm (\eg
\cite{cherubin2017website, de2020trafficsliver,henri2020protecting}).
Application layer defenses generally make strong assumptions that are often
unrealistic in practice, such as target websites implementing customized HTTP
protocols~\cite{ cherubin2017website} or only allowing attackers to observe
traffic at a single Tor entry
node~\cite{de2020trafficsliver,henri2020protecting}.  These defenses provide
less than 60\% protection against DNN-based attacks such as \textit{DF} and
\textit{Var-CNN}~\cite{sirinam2018deep,bhat2019var}.

In \textit{network layer obfuscation}, 
the defender inserts dummy packets into network traces to make website 
identification more challenging. First, early defenses~\cite{dyer2012peek,
cai2014cs,cai2014systematic} used constant-rate 
padding to reduce the information leakage caused by 
time gaps and traffic volume. However, these methods led 
to large bandwidth overhead ($> 150\%$). Second\abedit{ly}, more recent work 
reduces the overhead by inserting packets
in locations of the trace with large time gaps between packets 
(WTF-PAD~\cite{juarez2016toward}) or focusing on \abedit{the}
front portion of the trace, which has \abedit{been} shown to contain 
the most information (FRONT~\cite{gong2020zero}). 
However, WTF-PAD and FRONT only achieve $9\%$ and $28\%$ 
\abedit{protection} success respectively against strong WF attacks. 
Finally, supersequence defenses~\cite{wang2014effective,
nithyanand2014glove,wang2017walkie} try to 
find a \textit{supersequence}, which is a longer packet
trace that contains subsequences of 
different websites' traces. The strongest supersequence 
defense, Walkie-Talkie, achieves $50\%$ protection against
DNN attacks. Overall, against strong WF attacks, network layer obfuscation defenses
either induce extremely large overheads (> 100\%) or 
offer low protection (< 60\%). 

\para{Defenses via Adversarial Perturbation.}  Goodfellow \etal first
proposed evasion attacks against DNNs, where an attacker causes a DNN to
misclassify inputs by adding small, \textit{adversarial
  perturbations}~\cite{goodfellow2014explaining} to them. Such attacks have
been widely studied in many domains, \eg computer
vision~\cite{carlini,kurakin2016adversarial,
  chen2018ead,uesato2018adversarial,chen2018shapeshifter}, natural language
processing~\cite{zhang2020adversarial, ebrahimi2017hotflip}, and malware
detection~\cite{grosse2017adversarial, suciu2019exploring}. Recent WF
defenses~\cite{imani2019mockingbird,hou2020wf} use adversarial
perturbations to defeat DNN-based attacks.  

The challenge \htedit{facing} adversarial perturbation-based WF defense is that
adversarial perturbations are computed \abedit{with respect to} a given input, which in
the WF context is the full network trace. Thus, computing the adversarial
perturbation necessary to protect a network connection requires the defender
to know the \emph{entire trace} before the connection is even made.
This limitation renders \htedit{the defense} impractical to protect user
traces in real-time.

\bzedit{Mockingbird~\cite{imani2019mockingbird} suggests that we can use a
  database of recent traces to compute perturbations, and use them on new
  traces. Yet it is widely accepted in ML literature that adversarial
  perturbations are input specific and rarely
  transfer~\cite{moosavi2017universal}. We confirm this experimentally using
  Mockingbird's publicly released code: for the same website, perturbations
  calculated on one trace offer an average of $18\%$ protection to a
  different trace from the same website. We tested 10 pairs of traces per
  website, over 10 websites randomly sampled from the same dataset used by
  Mockingbird~\cite{imani2019mockingbird}.}

In a concurrent manuscript, Nasr \etal~\cite{nasr2020blind} propose \abedit{solving}
this limitation by precomputing {\em universal adversarial
  perturbations} for unseen traces, enabling the protection of live traces. 
However, an attacker aware of this defense can also compute these universal
perturbations, and adversarially train their models against them to improve
robustness. This countermeasure causes a significant drop in its protection to
$76\%$. We show that our defense outperforms~\cite{nasr2020blind}(\S\ref{sec:eval_prev}).

\begin{figure}
  \centering
  \includegraphics[width=1\linewidth]{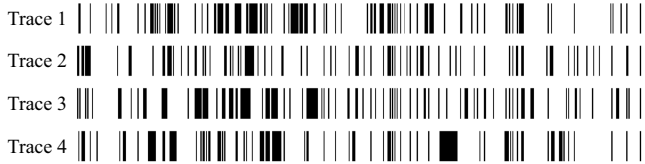}
  \caption{Example traces observed by a WF attacker
    when a Tor user $u$ visits the same 
    website.   Each sample trace ($x_u$) records the packet direction (black
    bar: incoming packet,
    white bar: outgoing packet) of the first 500 packets in this
    session. We see that while the Tor user visits the same website in these four
    sessions, the traces observed
    by the attacker vary across sessions. 
  }
  \label{fig:traces}
\end{figure}

\section{Problem and Threat Model}
\label{sec:methodology}
We consider the problem of defending against website fingerprinting
attacks. A user $u$ wishes to use the Tor network to visit
websites privately. An attacker, after eavesdropping on $u$'s traffic
and collecting a trace of $u$'s website visit ($x_u$),  attempts to
use $x_u$ to determine (or classify) the destination website that $u$ has just
visited.  To defend against such attacks, the defender seeks to inject
``obfuscation'' traffic into
the Tor network, such that  the traffic trace observed by the attacker  ($\hat{x_u}$) will
lead to a wrong classification result.


\para{Threat Model.} 
\label{subsec:threat_model}
We use the same threat model adopted by existing WF attacks and 
defenses.

\begin{packed_itemize} 
  
\item The attacker, positioned between the 
user and Tor guard nodes,  can only observe the 
user's network traffic but not
modify it.  Furthermore, the attacker can tap into user connections over 
a period of time and may see traffic on multiple website
requests.

\item We consider WF attacks that operate on packet
directions\footnote{ Since Tor's traffic is encrypted and padded to 
the same size~\cite{herrmann2009website}, only
packet directions and time gaps will leak information about 
the destination website.}.  This assumption is consistent with
many previous defenses~\cite{imani2019mockingbird,
  wang2017walkie,de2020trafficsliver,henri2020protecting,hou2020wf}. For each $u$'s website session, the attacker collects its trace
  $x_u$ as a sequence of packet directions (i.e., marking each
  outgoing packet as +1 and each incoming packet as
  -1). Figure~\ref{fig:traces} shows four examples when a Tor user visits the
  same website at different times.  While the Tor user visits the same website in these four sessions, the traces observed by the attacker vary across sessions.

\item We consider \textit{closed-world} 
WF  attacks where
the user only visits a set of websites 
known to the attacker. These attacks are 
strictly stronger than \textit{open-world}
attacks where the user can visit any website. The latter is far more
challenging for the attacker. 


\item We assume the attacker has full knowledge of the defense method deployed by
  the user and/or Tor (\ie access to source code). The attacker knows that
  the user traffic is protected by obfuscation, but does not know the
  run-time parameters (\ie user-side secret) used to construct the obfuscation
  traffic.

\end{packed_itemize}

\para{Defender Capabilities.}  We make two assumptions about the defender.

\begin{packed_itemize} \vspace{-0.06in}
\item The defender can 
actively inject dummy 
packets into user traffic. With cooperation from both the
user and a node in the Tor circuit (i.e., a Tor bridge\footnote{Tor bridges are often 
used to generate incoming packets for WF
defenses~\cite{juarez2016toward,imani2019mockingbird,gong2020zero,
  nasr2020blind}.}), the defender
can inject packets in both directions (outgoing and incoming). 

\item The defender has no knowledge of the WF classifier used by the
  attacker.


  \vspace{-0.06in}
\end{packed_itemize}

\para{Success Metrics.}  To successfully protect users against WF attacks, a
WF defense needs to meet multiple criteria. First, it should successfully
defend the most powerful known attacks (i.e., those based on DNNs), and reduce
their attack success to near zero. Second, it should do so without adding
unreasonable overhead to users' network traffic. Third, it needs to be
effective in realistic deployments where users cannot predict
the order of packets in a real-time network flow. Finally, a defense should
be robust to adaptive attacks designed with full knowledge of
the defense (i.e., with source code access).

Earlier in \S\ref{sec:defenses} we provide a detailed summary of existing
WF defenses and their limitations with respect to the above success metrics.

\section{ A New Defense using Adversarial Patches}
\label{sec:concept}
In this paper, we present \system{}, a new WF defense that exploits \abedit{the}
inherent weaknesses of deep learning classifiers to provide a highly
effective defense that protects network \abedit{traces} in real time, resists a variety
of countermeasures, and incurs low overhead compared to existing
defenses. 
In the following, we describe the key concepts behind \system{} and its design
considerations.  We present the detailed design
of \system{} later in \S\ref{sec:design}.

\vspace{-0.1in}
\subsection{Background: Adversarial Patches}
Our new WF defense is  inspired by the novel concept of {\em adversarial
patches}~\cite{brown2017adversarial} in computer vision.

Adversarial patches are a special form of artifacts, which when added to the
input of a deep learning classifier, cause an input to be misclassified.
Adversarial patches differ from adversarial perturbations in that patches are
both {\em input-agnostic} and {\em location-agnostic}, i.e., a patch causes
misclassification when applied to an input, regardless of the value of the
input or the location where the patch is applied \footnote{\abedit{We note that while input-dependent patches have been considered \cite{karmon2018lavan}, they are largely utilized and analyzed in an input agnostic context.}}.  Thus adversarial
patches are ``universal,'' and can be pre-computed without full knowledge of
the input. Existing works have already developed adversarial patches to
bypass facial recognition classifiers~\cite{wallace2019universal,
  wu2019making,song2018physical}.

In computer vision, an adversarial patch is formed as a fixed size pattern on
the image~\cite{brown2017adversarial}.  Figure~\ref{fig:adv_patch} shows an
example of an adversarial patch next to an example of adversarial
perturbation.  Given knowledge of the target classifier, one can search for
adversarial patches under specific constraints such as patch color or
intensity.  For instance, for a DNN model ($\mathbb{F}$), an targeted
adversarial patch $p_{adv}$ is computed via the following optimization:
\begin{equation}
   p_{adv}=\text{argmin}_{p} \E_{x\in \mathcal{X},l\in L}
   loss\left(\mathbb{F}(\Pi(p,x,l)), y_t\right) 
   \label{eq:adv_patch}
\end{equation}
where $\mathcal{X}$ is the set of training images, $L$ is a 
distribution of locations in the image, $y_t$ is the target label, 
and $\Pi$ is the function
that applies the patch to a random location of an image. 
This optimization is performed over all 
training images in order to make the patch effective 
across images.

\begin{figure}[t]
  \centering
  \includegraphics[width=0.45\textwidth]{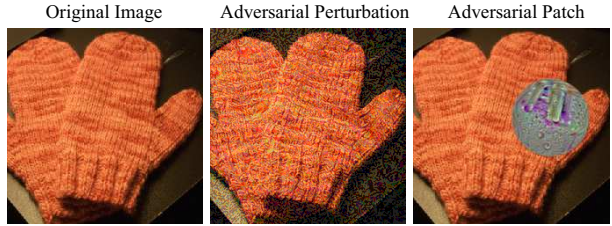}
  \caption{Sample images showing the difference between adversarial
    perturbations and adversarial patches.} 
  \label{fig:adv_patch}
\end{figure}
\subsection{Adversarial Patches as a WF Defense}

We propose to defend against WF attacks by adding {\em adversarial traffic
  patches} to user traffic traces, causing attackers to misclassify
destination websites.  To the best of our knowledge, our work is the first to
use adversarial patches to defend against WF attacks.


\vspace{1pt}
\noindent {\bf Trace-agnostic.} This is the key property of adversarial
patches and why they are a natural defense against WF attacks.  Given a user
$u$ and a website $W$, one can design a patch that works on any
network trace produced when $u$ visits $W$. We note that like adversarial
perturbations, patches can be computed as \abedit{the solution to a} constrained optimization problem. Empirical studies have not found any limitations in the number of
unique patches that can be computed for any targeted misclassification task.

Leveraging this property, the defender can pre-compute, for $u$, a set of
$W$-specific adversarial patches.  Once $u$ starts to visit the site $W$, the
defender fetches a pre-computed patch and injects, in real-time, the
corresponding dummy packets into $u$'s live traffic.  Furthermore, since
patches are built using diverse training data, they are inherently robust
against moderate levels of website updates and/or network dynamics. The
defender periodically re-computes new patches or when they detect significant
changes in website content and/or user networking environments.

\begin{figure*}
  \centering
  \includegraphics[width=0.82\linewidth]{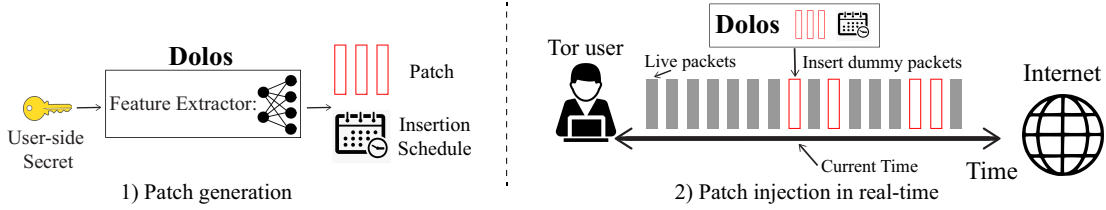}
  \vspace{-0.05in}
  \caption{Our proposed \system{} system that protects a user $u$ from WF 
    attacks. \emph{First}, \system{} precomputes a patch and its 
    packet insertion schedule  to protect $u$'s visits to website $W$, using $u$'s secret and a feature extractor. 
    \emph{Second}, when $u$ is visiting website $W$, \system{} defends
    user's traces in real-time by inserting dummy packets according to the precomputed patch and schedule. }
  \label{fig:defense_overview}
\end{figure*}

\subsection{Adversarial Patches for Network Traffic}
\label{subsec:challenges}

Here, we describe new design considerations that arise from applying
adversarial patches to network traffic traces. 

\htedit{For a specific user/website pair (user $u$, website $W$), the defender first uses
some sample traces of $u$ visiting $W$ to compute a patch $p$ for the pair. }At run-time, when $u$
initiates a connection to $W$, the defender fetches $p$ and follows the
corresponding insertion schedule to add dummy packets into $u$'s Tor
traffic. No original packets are dropped or modified.

Therefore, to generate
adversarial patches for traffic traces, we follow the optimization process
defined by Eq. (\ref{eq:adv_patch}), but change the $\Pi$ function for patch
injection.  \htedit{Note that the patches are generated for each specific user and
website pair ($u,W$). } Another key change is the ``perturbation budget'' (i.e., the
maximum changes to the input), which is now defined as the bandwidth overhead
introduced by those dummy packets.  Patches for images are often limited by a
small perturbation budget (e.g., $<5\%$), because bigger patches are
obvious to the human eye. Patches for traffic traces can support much larger
budgets.  For reference, Tor had deployed WTF-PAD as a WF defense, which
incurs $50+\%$ bandwidth overhead~\cite{sirinam2018deep,wtfpad2015takedown}.

\para{Strong Model Transferability of Patches.} 
Here {\em model transferability} refers to the well-known phenomenon that ML
classifiers trained for similar tasks share similar behaviors and
vulnerabilities, even if they are trained on different architectures or
training data~\cite{transfer2014}. Existing work has shown that when applied
to an input, an adversarial perturbation or patch computed for a given DNN
model will transfer across models~\cite{demontis2019adversarial,
  suciu2018does,petrov2019measuring}. In addition, a perturbed or patched
input will also transfer to non-DNN models such as SVM and random
forests~\cite{demontis2019adversarial,charles2019geometric}.  The level of
transferability is particularly strong for large perturbation sizes~\cite{demontis2019adversarial,petrov2019measuring}, as is
the case for adversarial patches for network
traffic.

Leveraging this strong model transferability, we can build a practical WF patch
without knowing the actual classifier used by WF attackers.  The
defender can compute adversarial patches using local WF models, and they
should succeed against attacks using other WF classifiers.

\subsection{Proposed: Secret-based Patch Generation}
A standard response by WF attackers to our patch-based defenses is to take patches generated
by the defense, build and label patched user traces, and use these traces to
(re)train an attack classifier to bypass the defense. In the adversarial ML
literature, this is termed ``adversarial training,'' and is regarded as the
most reliable defense against adversarial attacks, including adversarial
perturbations, adversarial patches, and universal perturbations.

In our experiments, we find that existing WF defenses fail under such attacks.
For~\cite{de2020trafficsliver}, protection drops to 55\%.  Noteworthy is the
fact that \cite{nasr2020blind}, which uses universal perturbations, is
particularly vulnerable to this attack. An attacker aware of the defense can
compute the universal perturbation themselves, and use it for adversarial
training. Our results show that the protection rate of \cite{nasr2020blind} under adversarial training drops dramatically
from 92\% to 16\%.

To resist adaptive attackers using adversarial training, we propose a
novel secret-based patch generation mechanism that makes it nearly impossible for the
attacker to reproduce the same patch as the defender.  \htedit{Specifically, the
defender first computes a user-side secret $\mathcal{S}$ based on private keys,
nonces, and website-specific identifiers}, and then uses it to
parameterize the optimization process of patch generation. The result is that
different secrets will generate significantly different patches. When
applied to the same network trace, the resulting patched traces will also
display {\em significantly disjoint} representations in both input and
feature spaces. Without access to the \shawn{user-side secret}, adversarial training using
patches generated by the WF attacker will have little effect, and \system{}
will continue to protect users from the WF attack.

\vspace{-0.05in}
\section{\system{}}
\label{sec:design}
In this section we present the design of \system{}, starting from an
overview, followed by detailed descriptions of its two key
components: patch generation and patch
injection.

\vspace{-0.05in}
\subsection{Overview}

Consider the task to protect a user $u$'s visits to website $W$.  \system{}
implements this protection by injecting an adversarial patch $p_{W,T}$ to
$u$'s live traffic when visiting $W$, such that when the
(defended) network
trace is analyzed by a WF attacker, its classifier will 
conclude that $u$ is visiting $T$ (a website different from $W$). Here
$T$ is a
configurable defense parameter.

\system{} includes two key steps: {\em patch generation} to compute an
adversarial traffic patch ($p_{W,T}$) and {\em patch injection} to inject a
pre-computed patch into $u$'s live traffic as $u$ is visiting $W$.  This is
also shown in Figure~\ref{fig:defense_overview}.

To generate a patch, \system{} inspects the WF feature space, and searches
for potential adversarial patches that can effectively ``move'' the feature
representation of a patched trace of $u$ visiting $W$ close to the feature
representation of the (unpatched) traces of $u$ visiting $T$.  When these two
feature representations are sufficiently similar, WF attacker classifiers will identify
the patched traces of $u$ visiting $W$ as a trace visiting $T$.

The above optimization can be formulated as follows: 
\begin{eqnarray}
   & p_{W,T}=\text{argmin}_{p} \E_{x\in X_W, x'\in X_T, s\in S}
     \mathbb{D}\left(\Phi(\Pi(p,x,s)), \Phi(x')\right)  \nonumber \\
   & \text{subject to} \;\; |p|\leq p_{budget}
   \label{eq:trafficpatch}
\end{eqnarray}
where $p_{budget}$ defines the maximum patch overhead, $X_W$ ($X_T$) defines
a collection of unpatched instances of $u$ visiting $W$ ($T$), $S$ defines
the set of feasible schedules to inject a patch into live traffic, and
$\Pi(p,x,s)$ defines the {\em patch injection} function that injects a patch
$p$ onto the live traffic $x$ under a schedule $s$.  Finally, $\Phi(\cdot)$
refers to the local WF feature extractor used by \system{} while $\mathbb{D}$
is a feature distance measure ($\ell_2$ norm in our implementation).
Figure~\ref{fig:apply} provides an abstract illustration of the patch and the
injection results.

\para{Randomized Design.} To prevent attackers from
extracting or reverse-engineering our patches, we design
\system{} to incorporate randomization into both patch
generation and injection. In \S\ref{subsec:optimization} and
\S\ref{subsec:patch_insert}, we describe how \system{} uses \shawn{user-side
secret} to configure $T$, $S$, $\Pi(p,x,s)$ and the optimization process of
Eq. (\ref{eq:trafficpatch}) to implement patches that are robust
against adaptive attacks. 

\para{Choosing $\Phi(\cdot)$.} As discussed earlier, \system{} does not
assume knowledge of the WF classifiers used by the attackers.  Instead,
\system{} operates directly on the {\em feature space} and uses a feature
extractor $\Phi$, basically a partial neural network trained on the same or
similar task. \system{} can train $\Phi$ locally or use a pre-trained WF
classifier from a trusted party (e.g., Tor).  Given an input $x$, \system{} uses the outputs of an
intermediate layer of $\Phi$ as $x$'s feature vector that quantifies distance
in the feature space.
A well-trained $\Phi$ can help \system{} tolerate web content
dynamics and can function well with a wider range of websites both known and
unknown~\cite{Rimmer2018}.


\begin{figure}[t]
  \centering
  \includegraphics[width=0.99\columnwidth]{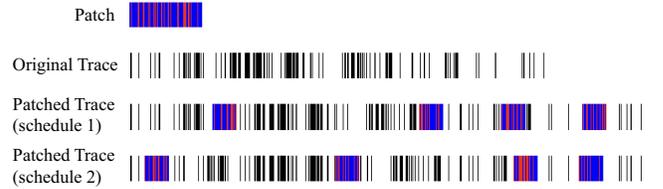}
  \caption{An abstract illustration of adversarial traffic patch and how it is
    injected into user traffic traces. Black/white bars mark the packet
    directions (out/in) of original packets, red/blue mark
    the (out/in) directions of dummy packets (i.e., the patch).  \system{}
    injects the patch into a live user trace following a specific 
    schedule. Here we show the results when the same patch is injected
    using two different schedules (1 \& 2) and random packet flipping
    (see \S\ref{subsec:patch_insert}).}
  \label{fig:apply}
\end{figure}

\subsection{Patch Generation}
\label{subsec:optimization}
The patch generation follows the optimization process defined by
eq. (\ref{eq:trafficpatch}). Here a novel contribution of \system{} is
to use a secret to configure the patch generation process, so
that the resulting patches are strictly conditioned on this secret. 
The means that when we apply patches generated with different secrets to the
same original trace, the resulting patched traces will be
significantly different in both input and WF feature spaces. The secret can
be a one way hash of a user's private key, a time-based nonce, and an
optional website specific modifier.  This secret allows \system{} to compute
multiple, distinct patches based on specific users, destination websites, and
time. A defender can periodically recompute patches with updated nonces to
prevent any longitudinal attacks that try to identify a common patch across
multiple connections to the same website.
This controlled randomization prevents attackers from observing the true
distribution of the patched traces over time or multiple traces.

\para{Parameterized Patch Generation.} When generating a patch, \system{}
uses a secret $\mathcal{S}$ to determine $T$ (the target website) and the exact
number of dummy packets to be injected, i.e., the length of the patch
$|p_{W,T}|$.

{\bf Choosing $T$: }  \system{} collects a large pool of candidate
target websites 
($400,000$ in our implementation) from the Internet. To protect user
$u$'s visit to $W$, \system{} uses $u$'s secret $\mathcal{S}$ to
``randomly'' select a website from the candidate pool, whose feature
representation is far from that of the original trace, i.e.,
$\mathbb{E}_{x\in X_W}\Phi(x)$ is largely dissimilar from
$\mathbb{E}_{x\in X_T}\Phi(x)$. 
For our implementation, we first calculate the feature distance
($\ell_2$ norm) between
$W$ and each candidate in our large pool,  identify the top $75^{\text{th}}$
percentile as a reduced candidate pool for $W$, and use the secret 
$\mathcal{S}$ as a random seed to select one from the reduced pool. 

{\bf Choosing $|p_{W,T}|$: } To further obfuscate the appearance of
our patches, we use the same secret $\mathcal{S}$ to select the patch length (i.e., the number of dummy packets to be injected into the user
traffic).  Specifically, given $p_{budget}$ (the maximum overhead
ratio), we ``randomly'' pick a patch length between
$(p_{budget}-\epsilon, p_{budget})$ where $\epsilon$ is a defense parameter
($0.1$ in our implementation). 

\para{Patch Optimization.}  When solving the patch optimization problem defined by
eq. (\ref{eq:trafficpatch}), 
we use Stochastic Gradient Descent (SGD)~\cite{shalev2014understanding} by batching samples 
from $X_W$. We note that in order to 
apply SGD, we need to relax the constraints on $p_{W,T}$ 
and allow it to lie in $[-1,1]^n$, \ie we allow 
for continuous values between $-1$ and $+1$ while 
optimizing. When the optimization ends, we quantize any negative value
as $-1$ and any non-negative value as $+1$. This method 
is widely used to solve discrete optimization problems in domains like natural 
language processing~\cite{wang2019natural,papernot2016crafting,
ren2019generating} and malware 
classification~\cite{kolosnjaji2018adversarial}, which avoids 
the intractability of combinatorial optimization~\cite{lee2004first}.
Finally, we note that the optimization also takes into account the
patch injection process (e.g., $\Pi(p,x,s)$), which we describe next in \S\ref{subsec:patch_insert}. 

\para{Managing Secrets.} Note that \system{} manages secrets and their
original components (private keys, nonces, website specific identifiers)
following the standard private key management
recommendations~\cite{barker2006recommendation}.  Secrets are only stored on
user's device and updated periodically.


\vspace{-0.05in}
\subsection{Patch Injection}
\label{subsec:patch_insert}

The patch injection component has two goals: 1) making the patch
input-agnostic and location-agnostic, in order to deploy our WF defense
on live traffic,  and 2) obfuscating the patched traces at run-time to prevent
attackers from detecting/removing patches from the observed
traces (to recover the original trace), or dismantling patches
via trace segmentation to reduce its coverage/effectiveness.

Existing solutions (used by adversarial patches for images) do not meet these
goals. They simply inject a given patch (as a fixed block) at any location of
the trace.  Under our problem context, an attacker can easily recognize the fixed pattern by observing
the user traces over a period of time.  Instead, we propose to combine a
segment-based patch injection method with run-time packet-level obfuscation.

\para{Segment-based Patch Injection.} A pre-computed patch $p$ is a sequence
of dummy packets (+1s and -1s) designed to be injected into the user's live
traffic $x$.  We first break $p$ into multiple segments of equal size $M_p$
(e.g., 30 packets), referred to as ``mini-patches.''  Each mini-patch is
assigned to protect a segment of $x$ of size $M_x$ (e.g., 100 packets).  We
configure the patch generation process to ensure that, within each segment of
$x$, the corresponding mini-patch stays as a fixed block. Patches are {\em
  location-agnostic}, so they will produce the same effect regardless of
location within the segment. Therefore, given $M_x$ and $M_p$, the injection
function $\Pi(p,x,s)$ in Eq. (\ref{eq:trafficpatch}) depends on $s$, an
injection schedule that defines the randomly chosen location of each
mini-patch within $x$'s segments (see Figure~\ref{fig:apply}). $S$
defines all possible sets of mini-patch locations.

An advantage of splitting up the patch into mini-patches is that it protects
against attackers trying to infer the website by searching for subsequences
of packets.  Our results confirm this hypothesis later in
\S\ref{subsec:detect}.

\para{Run-time Patch Obfuscation using Packet Flips.} When a single patch
$p$ is used to protect $u$'s visit to $W$ over some window of time, the same
$p$ could appear in multiple patched traces. While our segment-based
injection hides $p$ within the patched traces, it does not change $p$. Thus a
resourceful attacker could potentially recover $p$ using advanced trace
analysis techniques. To further elevate the obfuscation, we apply random
``packet flipping'' to make $p$ different in each visit to $W$.
Specifically, in each visit session, we randomly choose a small set of dummy
packets (in $p$) and flip their directions (out to in, in to out).  We ensure
that this random flipping operation is accounted by the patch generation
process (eq. (\ref{eq:trafficpatch})), so that it does not affect the patch
effectiveness.  Later in \S\ref{sec:counter}, we show that this random
flipping does deter countermeasures that leverage frequency analysis to
estimate $p$.

\para{Deployment Considerations.} At run-time, the user and Tor bridge follow
a simple protocol to protect traces. As soon as the user $u$ requests a
website $W$, \system{} sends the pre-computed patch $p_{W,T}$ (after random
flipping) and the current insertion schedule $s$ to the Tor bridge through an
encrypted and fixed length (padded) network tunnel. Next, the user and the
Tor bridge coordinate to send dummy packets to each other to achieve the
protection.

\section{Evaluation}
\label{sec:eval}
\vspace{-0.06in}
In this section, we perform a systematic evaluation of \system{} under
a variety of WF attack scenarios.  Specifically, we evaluate \system{} against i)  the state-of-the-art DNN WF attacks
whose classifiers use either the same or 
different feature
extractors of \system{} (\S\ref{sec:eval_dnn}), ii)  non-DNN WF attacks that use handcrafted
features (\S\ref{sec:eval_nondnn}).   Furthermore, we compare \system{} against existing WF defenses
 under these 
  attacks (\S\ref{sec:eval_prev}) and in terms of {\em information
    leakage},  which estimates the number of potential
  vulnerabilities facing any WF defense (including those not yet exploited by existing WF
  attacks) (\S\ref{sec:eval_leakage}).

Overall, our results show that \system{} is highly effective against
state-of-the-art WF attacks ($\geq 94.4\%$ attack protection at a 30\%
bandwidth overhead). It largely outperforms existing defenses in all
three key metrics: protection success rate, bandwidth overhead, and
information leakage.

Finally, under a simplified 2-class setting, we show that \system{} 
is {\em provably robust} with a sufficiently high bandwidth
overhead.  Our theory result and proof, which rely on the theory of optimal transport,
are listed in the Appendix.




\vspace{-0.06in}
\subsection{Experimental Setup}




\vspace{-0.05in}
\para{WF Datasets.} Our experiments use two well-known WF datasets:
\datasetsmall{} and \datasetlarge{} (see Table
\ref{tab:dataset}), which are commonly used by prior works for
evaluating  WF
attacks and defenses~\cite{bhat2019var,
holland2020regulator,sirinam2018deep,Rimmer2018,nasr2020blind}.  Both
datasets contain Tor users' website traces (as data) and their corresponding
websites (as labels). \datasetsmall{} was collected by Sirinam \etal around
  February 2016 and includes $86,000$ traces belonging to $95$ 
websites \cite{sirinam2018deep} in the Alexa top website
list~\cite{alexa}. \datasetlarge{}  was 
  collected by Rimmer \etal~\cite{Rimmer2018} in January 2017,
  covering $2$ million traces for visiting Alexa's top $900$
  websites. The two datasets have partial overlap in their
  labels. Following prior works on WF attacks and defenses, we pad each trace in the datasets into a fixed 
  length of $5000$.

  \begin{table}[h]
  \centering \vspace{-0.06in}
  \resizebox{0.48\textwidth}{!}{
    \begin{tabular}{|c|c|c|c|}
    \hline
    \textbf{Dataset Name} & \textbf{\# of Labels} & \textbf{\# of Training Traces} & \textbf{\# of Testing Traces} \\ \hline
    \datasetsmall & 95 & 76K & 10K \\ \hline
    \datasetlarge & 900 & 2M & 257K \\ \hline
    \end{tabular}
  }\vspace{-0.05in}
  \caption{Two WF datasets used by our experiments.}\vspace{-0.05in}
  \label{tab:dataset}
\end{table}




  \para{\system{} Configuration.}  Next we describe how we configure
  \system' feature extractor $\Phi$,  the patch injection parameters, and
  the user-side secret.  We build four feature extractors using two well-known 
  DNN architectures for web trace analysis, Deep CNN \shawn{(from the DF attack)} and
  ResNet-18 \shawn{(from the Var-CNN attack)}, and train them on the above two WF datasets.  We
  also apply standard adversarial patch
  training~\cite{bagdasaryan2020blind,nasr2020blind,rao2020adversarial}
  to fine-tune these feature
  extractors for 20 epochs, which helps increase the
  transferability of our adversarial 
  patches\footnote{We show the protection results of \system{} when
    using standard feature extractors without adversarial patch
    training in table~\ref{tab:transfer_nonrobust} in the Appendix. }. Table~\ref{tab:defense_config} lists the
  resulting $\Phi$s and we name them based on the model architecture
  and training dataset.

  \begin{table}[h]
  \centering
  \resizebox{0.48\textwidth}{!}{
    \begin{tabular}{|c|c|c|c|c|}
    \hline
    \textbf{Model Architecture $\mathbb{F}$} & \textbf{Training Data $\mathcal{X}$} & 
    \textbf{Feature Extractor $\Phi$}  \\ \hline
        Deep CNN (DF) & \datasetsmall  & \dfsmall  \\ \hline
        ResNet-18 (Var-CNN) & \datasetsmall  & \varsmall \\ \hline
        Deep CNN (DF) & \datasetlarge  & \dflarge \\ \hline
        ResNet-18 (Var-CNN) & \datasetlarge  & \varlarge \\ \hline
        \end{tabular}
      }\vspace{-0.05in}
  \caption{The four feature extractors ($\Phi$) used in our
    experiments, their model architecture and
    training dataset.}
  \label{tab:defense_config}
\end{table}

When injecting patches, we set the mini-patch length $M_p$ to $10$ and
the trace segment length $M_x$ to $M_p/R$, where $R$ represents the
bandwidth overhead of the defense ($0<R<1$).  We experiment with
different values of $M_p$ and find that its impact on the defense
performance is insignificant. Thus we empirically choose a value of
$10$.  By default, we set the packet flipping rate $\beta=0.2$, i.e.,
at run-time 
20\% of the dummy packets in each mini-patch will be flipped to a
different direction. The impact of $\beta$ on 
possible countermeasures is discussed later in \S\ref{subsec:detect}. 
\shawn{We use a separate dataset that contains traces from 400,000 different 
websites as our pool of target websites~\cite{Rimmer2018}. }

Finally, since our patch generation depends on the user-side secret,
we repeat each experiment $10$ times, each using a randomly
formulated user-side secret (per website). We report the average and
standard deviation values. Overall, the standard deviations are consistently low across our
experiments, i.e., $<1\%$ for protection success rate. 


\para{Attack Configuration. } We consider two types of WF attacks: 1)
non-DNN based and using handcrafted features, i.e.,  \textit{k-NN}~\cite{wang2014effective}, \textit{k-FP}~\cite{hayes2016k}, and 
\textit{CUMUL}~\cite{panchenko2016website}, and 2) DNN-based attacks, i.e., \textit{DF}~\cite{sirinam2018deep} and 
\textit{Var-CNN}~\cite{bhat2019var}.  
We follow the original implementations to implement
these attacks. 

Consistent with prior WF defenses~\cite{wang2017walkie,imani2019mockingbird,
gong2020zero,de2020trafficsliver,holland2020regulator},  we assume
that the attacker trains their classifiers using {\em defended traces}, i.e., those
patched using our defense.  To generate and label such traces,  the attacker downloads
\system{}  and runs it on their own traces when visiting a variety of
websites.  Here the attacker must input some user-side secret to run
\system{}, which we refer to as $\mathcal{S}_{attack}$. Since the
attacker has no knowledge of the user-side secret
$\mathcal{S}_{defense}$ being used by \system{} to
protect the current user $u$, we have $\mathcal{S}_{defense} \ne \mathcal{S}_{attack}$.  A relevant countermeasure by attackers is to enumerate many 
secrets when training the attack classifier, which we discuss later
in \S\ref{subsec: disrupt}. 

\htedit{
\para{Intersection Attacks.} We also consider {\em intersection attacks} in our evaluation of
\system{}.  Here the attacker assumes the victim visits the same
websites regularly and monitors the victim’s traffic for a
longer time period (e.g., days). With these information, the
attacker could make better inferences on user’s traces. We test
\system{} against the intersection attack used by a previous WF
defense~\cite{imani2019mockingbird} and find that the attack is
ineffective on \system{}.  More details about our experiments and
results can be found in the Appendix. }





\para{Evaluation Metrics. }  We test \system{} against various WF
attacks using the testing traces in the two WF datasets (see
Table~\ref{tab:dataset}).  We evaluate \system{} using 
three metrics: 1) \textit{protection success 
rate} defined as the WF attack's misclassification 
rate on the defended traces,  2) \textit{bandwidth overhead}  $R=\frac{\text{
patch length}}{\text{original trace length}} $, and 3) \textit{information 
leakage}, which  measures the amount of potential vulnerability of any
WF defense~\cite{cherubin2017bayes,li2018measuring}. 

We also examine the computation cost of \system{}. The average
time required to compute a patch 
is $19s$ on an Nvidia Titan X GPU and $43s$ on an eight core 
i9 CPU machine.

\begin{figure*}[t]
  \centering
    \begin{minipage}{0.32\textwidth}
    \centering
    \includegraphics[width=\textwidth]{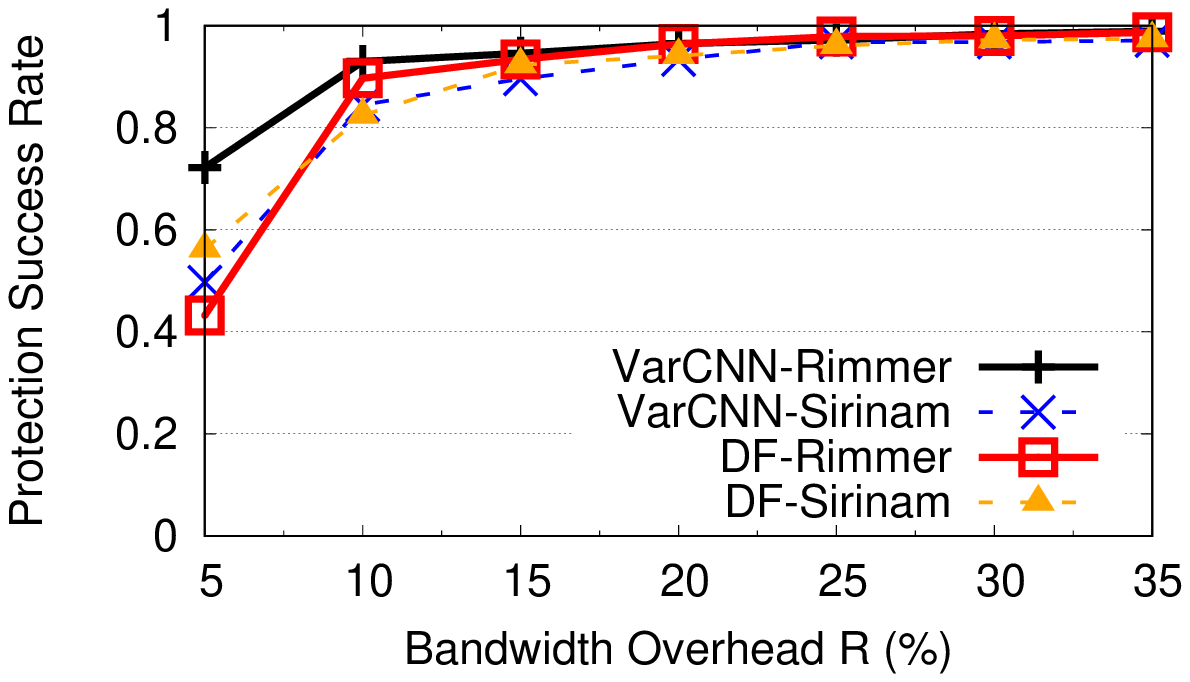}
    \caption{When against DNN-based attacks, \system{}'s
      protection success rate increases
      rapidly with its bandwidth overhead $R$.  \system{} achieves $>97\%$
      protection rate when $R$ reaches 30\%.  Assuming
      matching attack/defense. 
    }
  \label{fig:overhead}
  \end{minipage}
  \hfill
  \begin{minipage}{0.32\textwidth}
  \centering
  \includegraphics[width=\textwidth]{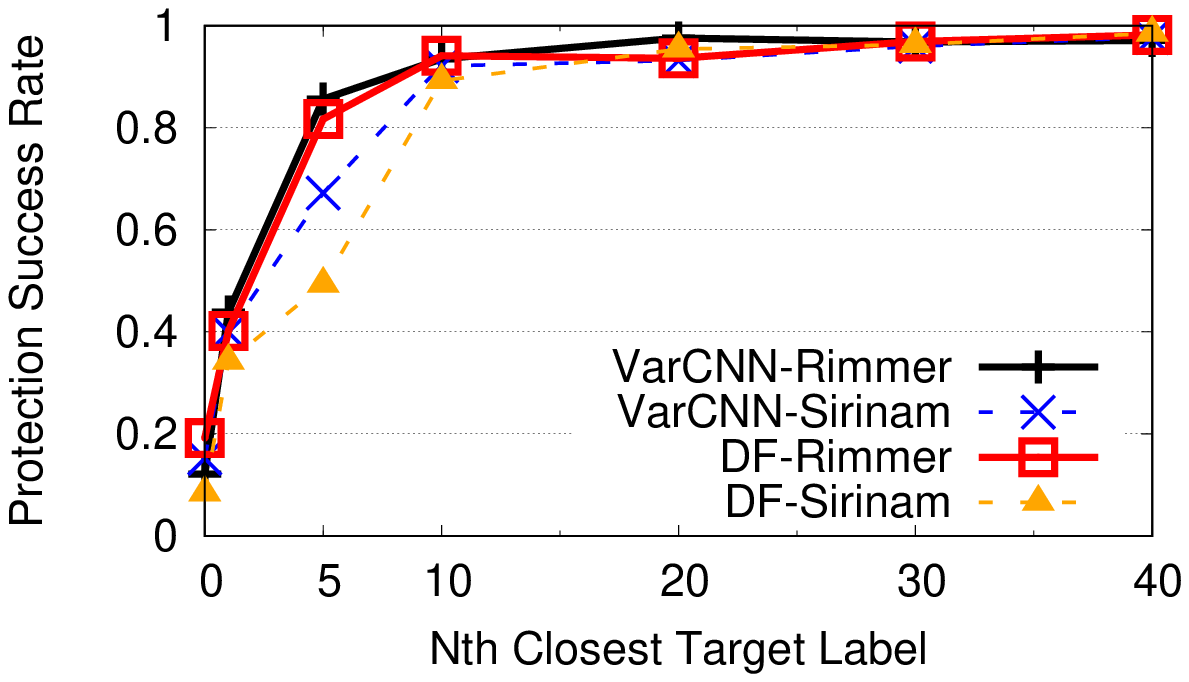}
  \caption{Worst-case
  analysis on the impact of key collision on \system,
  where the target feature representations $T_{attack}$ and $T_{defense}$ are
  $N^{th}$ nearest neighbors in the feature space. Assuming
      matching attack/defense. 
  }
  \label{fig:key_effect}
  \end{minipage}
  \hfill
  \begin{minipage}{0.32\textwidth}
  \centering
  \includegraphics[width=\textwidth]{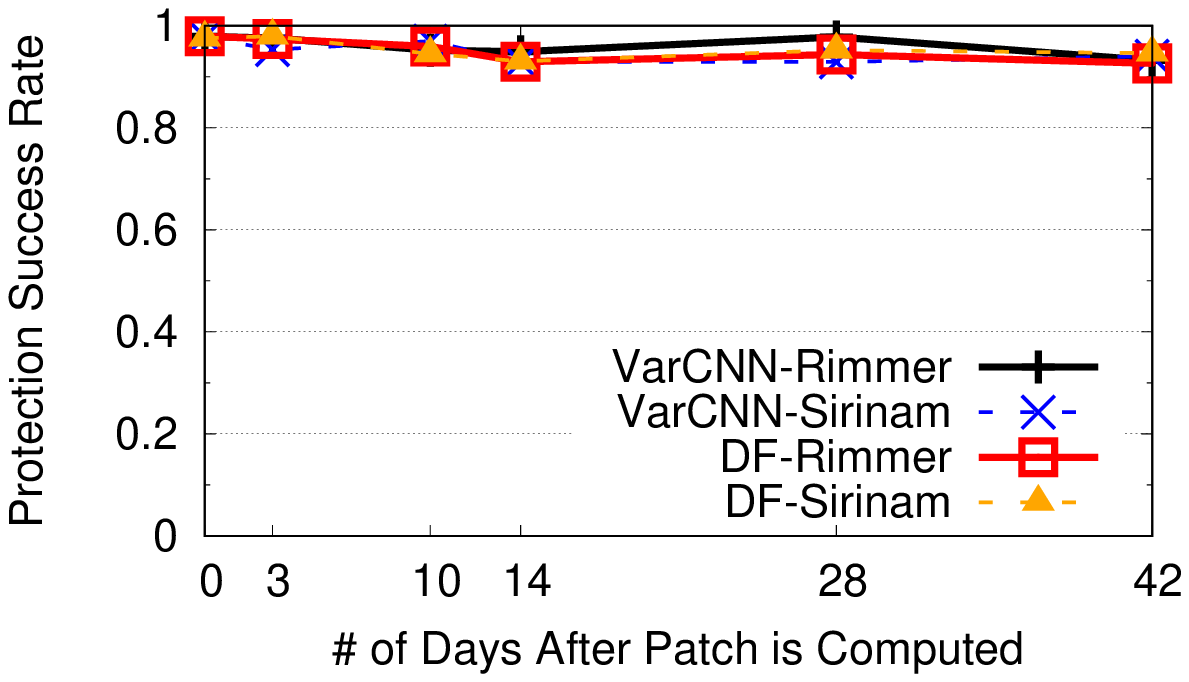}
  \caption{\system's effectiveness over time against fresh attacks after deploying a patch to protect $u$ visiting
    $W$ at day 0.  The same patch is able to resist attacks that train
    their classifiers on newly generated defended traces in
    subsequent days.} 
  \label{fig:drift}
  \end{minipage}
\end{figure*}

\begin{table*}[t]
  \centering
  \resizebox{0.79\textwidth}{!}{
\begin{tabular}{cc|ccc|cc}
\hline
\multirow{2}{*}{\textbf{Dataset}} &
                                    \multirow{2}{*}{\textbf{\begin{tabular}[c]{@{}c@{}}\system{}'s\\
                                                              Feature
                                                              Extractor\end{tabular}}}
                                  &
                                    \multicolumn{5}{c}{\textbf{\system's
                                    Protection
                                    Success Rate (\%) against WF
                                    Attacks 
                                    }} \\ \cline{3-7} 
 &  & \textbf{k-NN} & \textbf{k-FP} & \textbf{CUMUL} & \textbf{DF} & \textbf{Var-CNN} \\ \hline
\multirow{4}{*}{\begin{tabular}[c]{@{}c@{}}{\tt Sirinam}\end{tabular}} & \dfsmall & $98.2\pm0.2$ &
                                                                     $98.7\pm0.5$ & $97.3\pm0.4$ & $\mathbf{97.2\pm0.7}$ & $96.4\pm0.9$ \\
 & \dflarge & $97.4\pm0.5$ & $98.2\pm0.6$ & $96.3\pm0.6$ & $95.1\pm0.8$ & $97.2\pm0.5$ \\
 & \varsmall & $97.7\pm0.4$ & $98.0\pm0.6$ & $95.8\pm0.8$ & $94.4\pm0.6$ & $\mathbf{96.8\pm0.4}$ \\
 & \varlarge & $97.5\pm0.6$ & $98.9\pm0.3$ & $97.4\pm0.2$ & $95.9\pm0.8$ & $95.4\pm0.4$ \\ \hline
\multirow{4}{*}{\begin{tabular}[c]{@{}c@{}}{\tt Rimmer} \end{tabular}} & \dfsmall & $98.8\pm0.3$ & $97.6\pm0.5$ & $96.5\pm0.3$ & $96.3\pm0.5$ & $97.3\pm0.6$ \\
 & \dflarge & $97.0\pm0.5$ & $97.6\pm0.6$ & $97.8\pm0.6$ & $\mathbf{98.0\pm0.4}$ & $97.9\pm0.6$ \\
 & \varsmall & $98.2\pm0.4$ & $98.9\pm0.2$ & $97.3\pm0.5$ & $95.9\pm0.6$ & $97.3\pm0.4$ \\
 & \varlarge & $98.6\pm0.3$ & $98.1\pm0.4$ & $98.4\pm0.3$ & $96.8\pm0.5$ & $\mathbf{98.4\pm0.5}$ \\ \hline
\end{tabular}
  }\vspace{-0.0in}
  \caption{\system{}'s protection success rate against different WF
    attacks. The {\bf bold} entries are results under matching
    attack/defense. 
  }
  \label{tab:transfer}\vspace{-0.06in}
\end{table*}

\vspace{-0.06in}
\subsection{\system{} vs. DNN-based WF Attacks}
\label{sec:eval_dnn}
Our experiments consider three attack scenarios.
\begin{packed_itemize} \vspace{-0.06in}
  \item {\em Matching attack/defense}:  the attacker and \system{}
    operate on the same original trace data $\mathcal{X}$ 
(e.g., \datasetsmall) and use the same model architecture $\mathbb{F}$ (e.g.,
DF). \system{} trains its feature extractor $\Phi$ using model $\mathbb{F}$ and data $\mathcal{X}$,  while the attacker trains its attack
classifier using model $\mathbb{F}$ and  the {\em defended} version of
$\mathcal{X}$. 
\item {\em Mismatching attack/defense}: the
attacker and \system{} use different $\mathcal{X}$ and/or $\mathbb{F}$
when building their attack classifier and feature extractor,
respectively.

\item {\em Defense effectiveness over time}: \system{} starts to apply
  a patch
  $p$ to protect $u$'s visits to $W$ since 
  day 0; the attacker runs {\em fresh} WF attacks in the subsequent days by
  training attack classifiers on 
 defended traces freshly  generated in each day. 
\vspace{-0.06in}
\end{packed_itemize}

\para{Scenario 1: Matching Attack/Defense.}
Figure~\ref{fig:overhead} plots \system's protection success rate
against its bandwidth 
overhead, for each of the four ($\mathbb{F}, \mathcal{X}$)
combinations listed in Table~\ref{tab:defense_config}. The standard
deviation values are small ($<1\%$) and thus omitted from the figure
for clarity.

We see that
\system{} consistently achieves $97\%$ or higher protection rate when the
defense 
overhead $R \geq 30\%$. This result confirms that the adversarial
patches, when controlled via user-side secrets,  can effectively break
WF attacks trained on defended traces. Note that \system{} is even more
  effective against WF  attacks whose classifiers are 
  trained on original (undefended) traces, i.e., 98\% protection success rate with a 15\% overhead. 
For the rest of the paper, we use 
$R=30\%$ as the default configuration.

\vspace{1pt} {\bf Likelihood of Secret Collisions:}   In the above
experiments, we randomly select the secret pair ($\mathcal{S}_{attack},
\mathcal{S}_{defense}$) and show that in general, as long as $\mathcal{S}_{attack}\neq
\mathcal{S}_{defense}$, \system{} is highly effective against WF
attacks.  Next, we also run a {\em worst-case} analysis
that looks at cases where the combination of $\mathcal{S}_{attack}$ and $\mathcal{S}_{defense}$ leads to
heavy collision in the WF feature space. That is, the patches
generated by $\mathcal{S}_{attack}$ and $\mathcal{S}_{defense}$ will move the feature
representation of the original trace to the target feature
representations of website $T_{attack}$ and $T_{defense}$,
respectively, but  the two targets are close in 
the feature space (with respect to the $\ell_2$ distance).  Here we
ask the question: {\em how
``close'' do $T_{attack}$ and $T_{defense}$ need to be in order to break our defense}? 

We answer this question in Figure~\ref{fig:key_effect} by plotting
\system's protection success rate when $T_{attack}$ is the $N^{th}$ nearest
label to $T_{defense}$ in the feature space ($N\in [1,40]$).  Here we show the
result for each of the four ($\mathbb{F}, \mathcal{X}$)
combinations.  We see that as long as $T_{attack}$ is beyond the top
$20^{th}$ nearest neighbors of $T_{defense}$,  \system{} can maintain $>96\%$
protection success rate.   Since our pool of target websites is
very large (400,000 websites), the probability of an attacker finding a
$K_{attack}$ that weakens our defense is very low
($p_{\text{bad}}\approx\frac{20}{400,000}=5 \times 10^{-5}$).  Later
in 
\S\ref{sec:counter}, we show that even when the attacker trains their classifiers on defended traces produced
by a large number of secrets, the impact on \system{} is still minimum.

\begin{table*}[t]
  \centering
  \resizebox{0.8\textwidth}{!}{
  \begin{tabular}{ccc|cccccccc}
  \hline
  \multirow{2}{*}{\textbf{Dataset}} & \multirow{2}{*}{\textbf{Defense Name}} & \multirow{2}{*}{\textbf{\begin{tabular}[c]{@{}c@{}}Bandwidth \\ Overhead\end{tabular}}} & \multicolumn{6}{c}{\textbf{Protection Success Rate Against WF Attacks}} \\ \cline{4-9} 
   &  &  & \textbf{k-NN} & \textbf{k-FP} & \textbf{CUMUL} &
                                                            \textbf{DF}
                         & \textbf{Var-CNN} & \textbf{Worst Case
                                              } \\ \hline
  \multirow{4}{*}{\textbf{\datasetsmall}} & WTF-PAD & 54\% & 87\% & 56\% & 69\% & 10\% & 11\%  & 10\% \\
   & FRONT & 80\% & 96\% & 68\% & 72\% & 31\% & 34\%  & 31\% \\
   & Mockingbird & 52\% & 94\% & 89\% & 91\% & 69\% & 73\% &  69\% \\
   & \textbf{\system{}} & \textbf{30\%} & \textbf{98\%} & \textbf{99\%} & \textbf{97\%} & \textbf{96\%} & \textbf{95\%} & \textbf{95\%} \\ \hline
  \multirow{6}{*}{\textbf{\datasetlarge}} & WTF-PAD & 61\% & 84\% & 58\% & 72\% & 14\% & 11\% & 11\% \\
   & FRONT & 72\% & 97\% & 62\% & 68\% & 37\% & 39\% &  37\% \\
   & Mockingbird & 57\% & 96\% & 87\% & 90\% & 71\% & 79\% &  71\% \\
  & Blind Adversary & 11\% & - & - & - & - & $76\%^*$ &  76\% \\
   & \textbf{\system{}} & \textbf{10\%} & \textbf{95\%} & \textbf{92\%} & \textbf{93\%} & \textbf{87\%} & \textbf{92\%}  & \textbf{87\%} \\
   & \textbf{\system{}} & \textbf{30\%} & \textbf{99\%} & \textbf{98\%} & \textbf{98\%} & \textbf{97\%} & \textbf{98\%}  & \textbf{97\%} \\ \hline
  \end{tabular}
  }\vspace{-0.06in}
  \caption{Comparing bandwidth overhead and protection success rate of WTF-PAD, FRONT, 
  Mockingbird, Blind Adversary, and \system{}. $^*$ we take the number from their original paper~\cite{nasr2020blind} as the authors have not released their source code. }\vspace{-0.06in}
  \label{tab:compare}
\end{table*}

\para{Scenario 2: Mismatching Attack/Defense.} We now consider the
more general scenario where the attacker and \system{} use different
$\mathcal{X}$ and/or $\mathbb{F}$ to train their classifiers and
feature extractors, respectively.  Here we consider two existing DNN-based WF attacks, DF and
Var-CNN, trained on \datasetsmall{} or \datasetlarge{},  and configure
\system{} to use one of the four feature extractors listed in
Table~\ref{tab:defense_config}.   Using the test data of
\datasetsmall{} and \datasetlarge{}, we evaluate \system{} against
these WF attacks (DF and Var-CNN), and list \system's protection success rate in 
Table~\ref{tab:transfer}.   As reference, we also include the results
of matching attack/defense (scenario 1), marked in bold.

Overall, \system{} remains highly effective ($>94\%$ protection rate)
against attacks using different training data and/or model
architecture.   This shows that our proposed adversarial patches generated from a local feature extractor
can successfully {\em transfer} to a variety of attack classifiers.

\para{Scenario 3: Defense Effectiveness Over Time.} Next, we
evaluate \system{} against freshly generated attacks over
time.  Here \system{} computes and deploys a patch $p$ to protect
$u$'s visits to $W$ at day 0; the attacker continues to run WF attacks in the
subsequent days, and each day, they train the attack classifier using
defended traces generated on the current day.  We use this experiment
to examine the robustness of \system's patches under web content
dynamics and network dynamics, also referred to as concept drift.


Our experiment uses the concept drift dataset provided by Rimmer
\etal~\cite{Rimmer2018}, which was collected along with \datasetlarge. 
This new dataset consists of $200$ websites (a subset of 900 websites
in \datasetlarge), repeatedly collected over a six week period (0 day, 3 days, 10 days, 
14 days, 28 days, 42 days).  We run \system{} using each of the four
feature extractors to produce a patch at day 0.  The attacker uses the
Var-CNN classifier and trains the classifier using fresh defended traces
generated on  day 3,
10, 14, 28 and 42.  We see that the protection success rate remains consistent
($>93\%$) over
42 days.


\secspace
\subsection{\system{} vs. non-DNN Attacks}
\label{sec:eval_nondnn}

While \system{} targets the inherent vulnerability of
DNN-based WF attacks, we show that the adversarial patches produced by
\system{} are also highly effective against non-DNN based WF attacks.
Here we consider three most-effective non-DNN attacks:
k-NN~\cite{wang2014effective}, k-FP~\cite{hayes2016k}, and CUMUL~\cite{panchenko2016website}. Table~\ref{tab:defense_config} lists the protection
success rate of \system{} under these attacks,  using each of the four 
different feature extractors.  Our results align with existing
results: adversarial patches designed
for DNN models also transfer to non-DNN models~\cite{demontis2019adversarial,charles2019geometric}.


\secspace
\subsection{Comparison with Previous WF Defenses}
\label{sec:eval_prev}

Table~\ref{tab:compare} lists the performance of \system{} and four 
state-of-the-art defenses (WTF-PAD~\cite{juarez2016toward}, 
FRONT~\cite{gong2020zero}, Mockingbird~\cite{imani2019mockingbird}, Blind Adversary~\cite{nasr2020blind} as described in \S\ref{sec:back}). 
We evaluate them against five attacks on the two WF datasets. For \system{}, we use \varlarge{} as
the local feature extractor. 

\para{WTF-PAD.} WTF-PAD is reasonably effective against
traditional ML attacks, but performs poorly against any DNN
based attack, \ie protection success rate drops to 10\%. 
These findings align with existing 
observations~\cite{juarez2016toward,sirinam2018deep}. 

\para{FRONT.} FRONT is effective against non-DNN attacks but fails against DNN attacks. 
In our experiments, FRONT induces
larger bandwidth overhead than reported in the 
original paper~\cite{gong2020zero}. The discrepancy is 
because FRONT's overhead is dataset dependent and
a separate paper~\cite{holland2020regulator} reports the same overhead as ours when 
applying FRONT to \datasetsmall. 

\para{Mockingbird.} Mockingbird is effective
against non-DNN attacks and reasonably
effective against DNN attacks (71\% protection success). 
As stated earlier, the biggest drawback of Mockingbird is that
the defense needs the access to the full 
trace beforehand, making it unrealistic
to implement in the real-world. 

\para{Blind Adversary.} We are not able to obtain
the source code of Blind adversary at the time of 
writing. In the original paper, Blind Adversary 
is evaluated on the same dataset (\datasetlarge) 
using same robust training 
technique. However, the authors report the results for 
11\% overhead under the \emph{white-box} setting. 
Using the same setting and similar overhead, \system{} 
achieves $92\%$ protection success rate
whereas Blind Adversary achieves $76\%$ protection.

\begin{figure*}[t]
  \centering
  \hfill
  \begin{minipage}{0.32\textwidth}
    \centering
    \includegraphics[width=1\columnwidth]{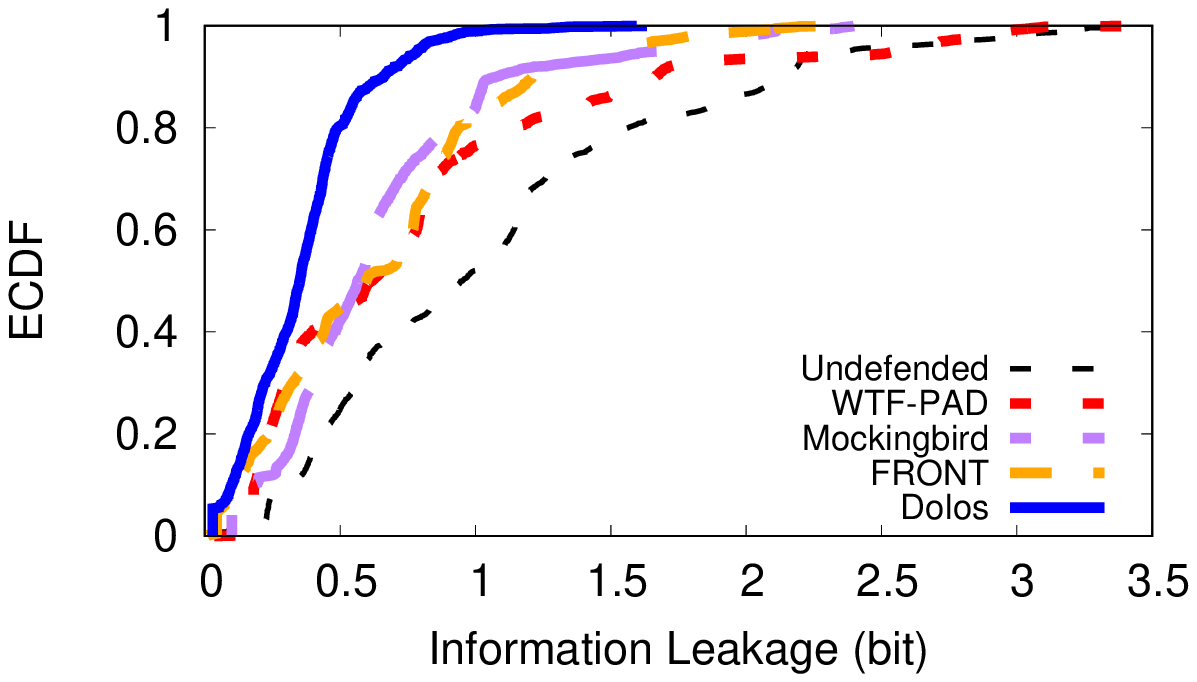}
  \caption{The ECDF of information leakage on hand crafted 
  features from WeFDE. }
  \label{fig:infoleakage}
  \end{minipage}
  \hfill
  \begin{minipage}{0.32\textwidth}
    \centering
    \includegraphics[width=1\columnwidth]{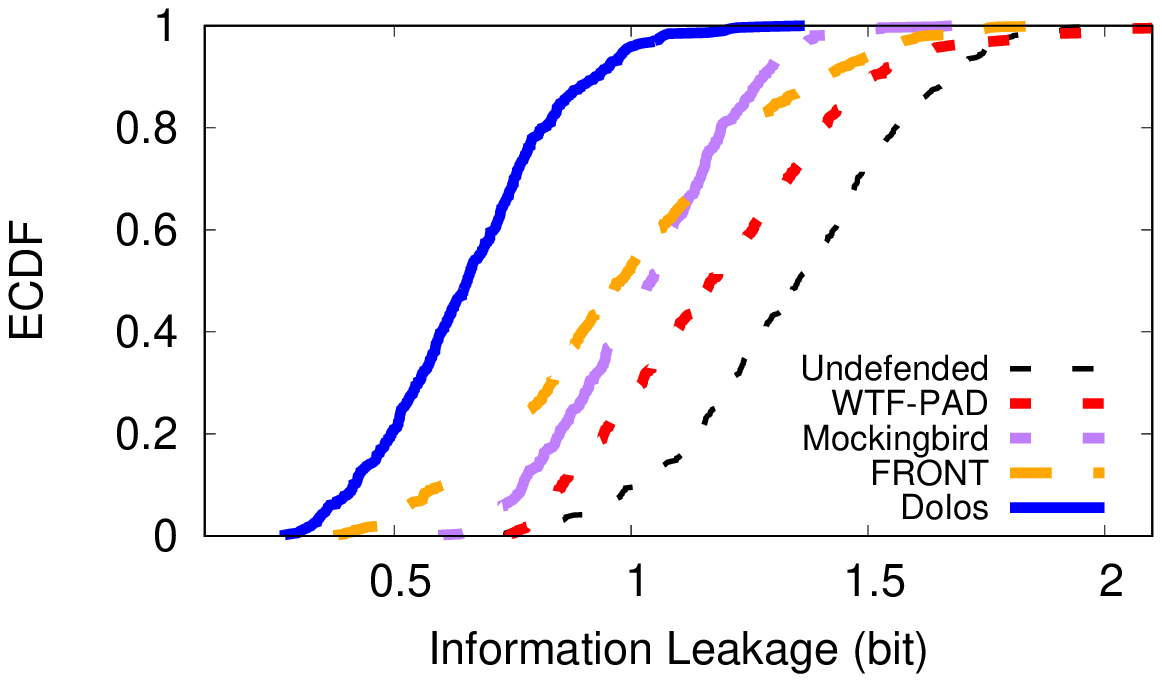}
  \caption{The ECDF of information leakage on features from models trained on defended traces. }
  \label{fig:infoleakage_dnn}
  \end{minipage}
  \hfill
  \begin{minipage}{0.32\textwidth}
    \centering
    \includegraphics[width=1\textwidth]{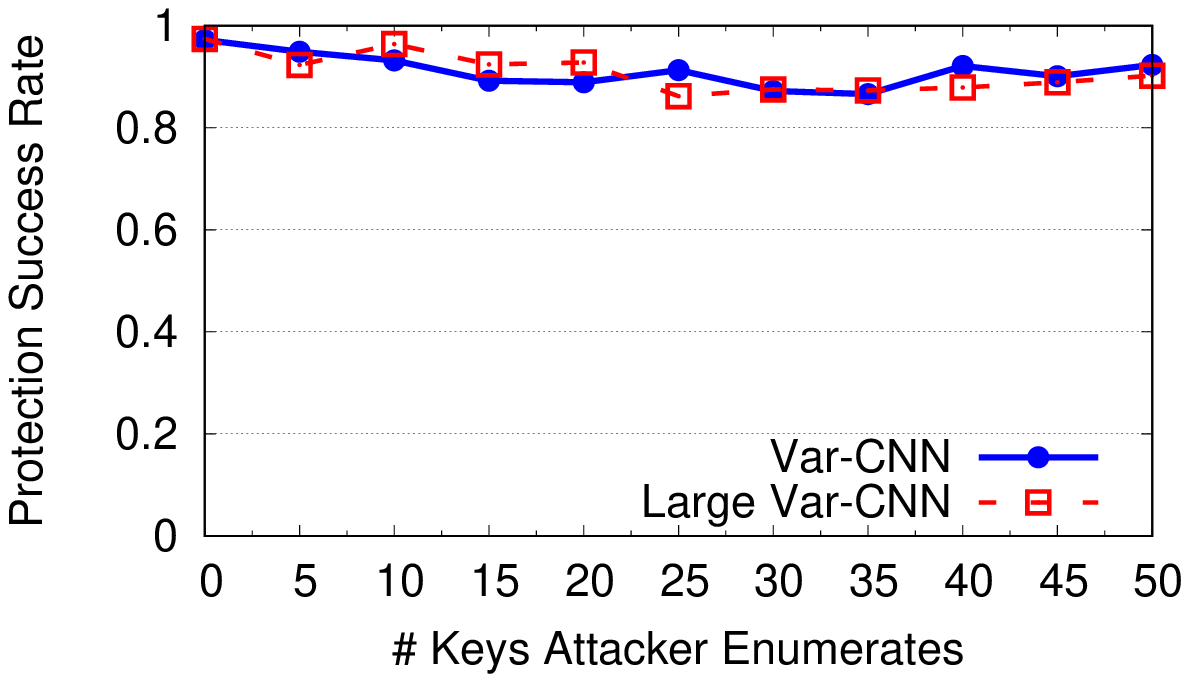}
    \caption{Protection performance drops slightly as attacker trains on traces defended by increasing numbers of secrets. }
    \label{fig:protected_train}
  \end{minipage}

\end{figure*}

\secspace
\subsection{Information Leakage Analysis}
\label{sec:eval_leakage}
Recent works~\cite{cherubin2017bayes,li2018measuring}
argue that WF defenses need to be evaluated beyond protection success
rate because existing WF attacks may not
show the hidden vulnerability of 
the proposed defense. Li \etal~\cite{li2018measuring}
proposes an information 
leakage estimation framework (WeFDE)
that measures a given defense's information 
leakage on a set of features. We measure \system{}
information leakage following the WeFDE framework. 

WeFDE computes the mutual information of trace 
features between different classes. 
A feature has a high information leakage 
if it distinguishes
traces between different websites. 
However, we cannot directly apply such analyzer to measure
the leakage of \system{} because \system{} does not
seek to make traces from different websites
indistinguishable from each other. Attacker can separate the traces defended by \system{} but has no
way of knowing which websites the traces belong to. 

Thus to measure the information
leakage of \system, we need to obtain the overall distribution
of defended traces agnostic to secrets, \ie defended
traces using all possible secrets. We approximate this distribution
using traces generated with a large number of secrets
and feed the aggregated traces to WeFDE for information 
leakage analysis. 

We use all websites from \datasetlarge{}. For each website, we 
use \dfsmall{} and $80$ different secrets to generate defended traces. 
We find that enumerating more than $80$ secrets has limited 
impact on the information leakage results. 
We aggregate the traces together before feeding into WeFDE. 
We compare \system{} with three 
previous state-of-the-art defenses, WTF-PAD, FRONT, and MockingBird. 
We measure the leakage (bits) on two sets of features: 1) hand crafted 
features from the WeDFE paper~\cite{li2018measuring}, and 2) DNN features from 
a model trained on the defended traces of a given defense. 

\para{Leakage on Hand Crafted Features.} We first 
measure the information leakage on the default
set of $3043$ hand crafted features 
from WeFDE. This feature set 
covers various categories, \eg packet ngrams, 
counts, and bursts. Figure~\ref{fig:infoleakage}
shows the empirical cumulative distribution 
function (ECDF) of information 
leakage across the features. The curve of 
\system{} increases much faster than the 
curves of other defenses. 
For \system, no feature leaks more than $1.7$ bits of 
information, which is lower 
than the minimum leakage of other defenses. 

\para{Leakage on Features Trained on Defended 
Traces.} For each defense, we measure
the information leakage on the feature space
of Var-CNN models trained on the defended traces. For \system, the
defended traces used by the attacker are generated by a different
secret. 
Figure~\ref{fig:infoleakage_dnn}
shows the ECDF of the leakage across all the features. 
Overall, this feature set leaks more information
than the hand crafted features in Figure~\ref{fig:infoleakage}. 
Again the curve of \system{} increases faster than 
the curves of other defenses. The gap between \system{}
and other defenses is larger in this set of features
showing \system{} is more resilient against models
trained on defended traces due to randomness introduced 
by the secret.

\vspace{-0.1in}
\section{Countermeasures}
\label{sec:counter}

\begin{table*}[t]
  \centering
  \begin{minipage}{0.32\textwidth}
    \centering
    \resizebox{1\textwidth}{!}{
    \begin{tabular}{|c|c|c|}
    \hline
    \multirow{2}{*}{\textbf{\begin{tabular}[c]{@{}c@{}}Ratio of\\ Packets\end{tabular}}} & \multicolumn{2}{c|}{\textbf{Protection Success Rate When}} \\ \cline{2-3} 
         & \textbf{Drop Packets} & \textbf{Flip Packets} \\ \hline
    0.05 & 97\%                  & 97\%                  \\ \hline
    0.1  & 96\%                  & 97\%                  \\ \hline
    0.2  & 95\%                  & 96\%                  \\ \hline
    0.3  & 96\%                  & 96\%                  \\ \hline
    0.5  & 96\%                  & 96\%                  \\ \hline
    \end{tabular}
    }
      \caption{Protection success rate remains high as attacker drops or flips a portion of packets before classification. }
    \label{tab:dropflip}  
    \end{minipage}
    \hfill
  \begin{minipage}{0.32\textwidth}
    \centering
    \resizebox{1\textwidth}{!}{
      \begin{tabular}{|c|c|}
      \hline
    \multirow{2}{*}{\textbf{\begin{tabular}[c]{@{}c@{}}Ratio of Packets \\ Trimmed\end{tabular}}} &
      \multirow{2}{*}{\textbf{\begin{tabular}[c]{@{}c@{}}Protection Success \\ Rate\end{tabular}}} \\          &      \\ \hline
      0.1 & 97\% \\ \hline
      0.2 & 97\% \\ \hline
      0.3 & 96\% \\ \hline
      0.5 & 96\% \\ \hline
      \end{tabular}
    }
    \caption{Protection success rate remains high as attacker trim a portion of packets at the end of each trace before classification. }
  \label{tab:trimpackets}   

  \end{minipage}
    \hfill
    \begin{minipage}{0.32\textwidth}
      \centering
    \resizebox{1\textwidth}{!}{
      \begin{tabular}{|c|c|}
      \hline
      \multirow{2}{*}{\textbf{\begin{tabular}[c]{@{}c@{}}\# of Epochs Trained\\ with Adversarial Training\end{tabular}}} & \multirow{2}{*}{\textbf{\begin{tabular}[c]{@{}c@{}}Protection Success \\ Rate\end{tabular}}} \\
       &  \\ \hline
      20 & 96\% \\ \hline
      40 & 95\% \\ \hline
      100 & 94\% \\ \hline
      200 & 95\% \\ \hline
      \end{tabular}
    }
    \caption{Protection success rate remains high when 
    transferring to increasing more robust models.  }
  \label{tab:robust}   
  \end{minipage}
\end{table*}

In this section, we explore additional countermeasures that could be launched
by attackers with complete knowledge of \system{}. We consider three classes
of countermeasures: detecting \system{} patches, preprocessing inputs to
disrupt \system{} patches, and boosting WF classifier robustness.  Unless
otherwise specified, experiments in this section run VarCNN-based attacks on
the \datasetlarge{} dataset, and the defender uses the \dfsmall{} feature
extractor to generate patches (see \S\ref{sec:eval}).



\subsection{Detecting \system{} Patches}
\label{subsec:detect}

An attacker can apply data analysis techniques to detect the presence of
patches in network traces. Detection can lead to possible identification of
patches and their removal. 

\para{Frequency Analysis.} An attacker who observes multiple visits to the
same website by the same user might identify defender's patch sequences using
frequency analysis, if the same patch is applied to multiple traces over a
period of time.  In practice, this frequency analysis might be challenging
because: the location of the patch is randomized, \system{} randomly flips a
subset of the dummy packets each time the patch is applied, and packet
sequences in patches might blend in naturally with unaltered network traces.

We test the feasibility of this countermeasure. We assume the attacker has
gathered network traces from $100$ separate visits by the same user to a
single website. For each trace, the attacker enumerates all packet sequences
with the length of the mini-patch (known to attacker). To address the random
flipping, the attacker merges packet sequences that have a Hamming distance
smaller than the flipping ratio (known to the attacker).  This produces a set
of packet sequences for each trace.  The sequence of each mini-patch should
appear in every set.

As the flip ratio increases, however, packet sequences from the patch start
to blend in with common packet sequences found frequently in benign
(unpatched) network traces. Patch sequences can thus blend in with normal
sequences, making their identification and removal difficult.

For example, for each website in \datasetlarge, we take $100$ original
network traces and perform frequency analysis. With a flip ratio of $0.2$, an
website has on average 45\% of its packet sequences showing up in every set
as false positives that look like potential patches.  An aggressive attacker 
can remove \emph{all} such high frequency packet sequences
before classifier  inference. 
Removing these high frequency packet sequences means the attacker's
classifier accuracy is reduced down to 7\%.
We perform this test using different values of flip ratio $\beta$ and show
the results in Table~\ref{tab:frequency}. 
success rate against a normal Var-CNN attack and 
frequency analysis attack. When the flip ratio is
$\geq 0.2$, frequency analysis countermeasure offer
no benefit against our defense. 

\begin{table}[t]
  \centering
    \resizebox{0.43\textwidth}{!}{
    \begin{tabular}{|c|c|c|c|}
    \hline
    \multirow{2}{*}{\textbf{Flip Ratio}} & \textbf{Protection Success } &
                                                                         \textbf{Success of} \\
 & \textbf{Var-CNN Attack} & \textbf{Frequency Analysis} \\ \hline
    0   & 98\% & 92\%   \\ \hline
    0.1 & 96\% & 26\%  \\ \hline
    0.2 & 96\% & 5\%  \\ \hline
    0.3 & 92\% & 0\% \\ \hline
    \end{tabular}
    }
    \caption{Impact of flip ratio on frequency analysis attack.}
  \label{tab:frequency}   
\end{table}

\para{Anomaly Detection.} We also consider attackers that using traditional
anomaly detection techniques to distinguish patches from normal packet
sequences.

We compute 80 distinct patches using \dfsmall{} for each of the 900 websites
in \datasetlarge{}, for a total of 72,000 patches. We compare these patches to 
72,000 natural packet sequences of the same length randomly chosen from
original traces (with random offsets). \textit{First}, we run a
2-means clustering on the features space of sequences (features extracted by
\dfsmall). The resulting clusters contain $47\%$ and $53\%$ patches
respectively and fail to distinguish patches from
normal sequences. \textit{Second}, we also try to separate patches using
supervised training. We train a DF classifier on packet sequences of
patches and natural packets. The classifier achieves $58\%$ accuracy,
only slightly outperforming random guessing.


\subsection{Preprocessing Network Traces}
A simple but effective approach to defeat adversarial patch
is to transform the input before using them for 
both training and inferencing~\cite{carlini2017adversarial,
feinman2017detecting}. In the case of WF attacks, 
we consider 3 possible transformations: i) add ``noise'' by randomly
flipping packets, ii) add ``noise'' by randomly dropping packets, iii)
truncating the network traces after the first $N$ packets. 
Note that the attacker is processing the traces 
locally and does not modify any packets in network. 

Our tests show that none of these transformations impact our defense in any
meaningful way. Flipping random packets in the trace degrades the
classification accuracy of the attacker classifier by $22\%$, but \system{}
remains successful $> 96\%$ (Table~\ref{tab:dropflip}).  Next,
Table~\ref{tab:dropflip} shows randomly dropping packets from the trace
decreases protection success rate by at most $2\%$, but more significantly,
degrades attacker's classification accuracy (by $32\%$).  Finally, truncating
the trace degrades the attacker's classification accuracy but \system{}
protection success remains $> 94\%$ (Table~\ref{tab:trimpackets}).

\vspace{-0.1in}
\subsection{More Robust Attack Classifiers}
\label{subsec: disrupt}
Next, we evaluate the feasibility of techniques to improve the robustness of
attacker classifiers against adversarial patches.

\para{Training on Multiple Patches.} In \S\ref{sec:eval}, the attacker
trained their classifier on traces protected by patch generated using a
single secret, and failed to achieve high attack performance. Here, we
consider a more general adversarial training approach, that trains the
attacker's model against patched traces generated from {\em multiple}
distinct secrets. Training against multiple targets gets the model closer to
a more complete coverage of the space of potential adversarial patches.


The attacker uses \system{} source code (with \dfsmall{} feature extractor)
to generate defended traces using $N$ randomly selected keys for each website
in \datasetlarge{} and trains a Var-CNN classifier. On the defender side, the
traces is protected using \dfsmall{} feature extractor but a different
key\footnote{We do not consider the case where the user and attacker's
  secrets match, since its probability is extremely small (\ie
  $\leq 50 / 400K$ in cases we tested).}.  Figure~\ref{fig:protected_train}
shows that as the attacker trains against more patched traces generated from
different secrets, there is a small gain in robustness by the attacker's
model. At its lowest point, the efficacy of \system{} patches drops to
$87\%$. Across all of our countermeasures tested, this is the most
effective. 


\para{Robust Attack Classifier.} In adversarial patch training, a model is iteratively retrained on 
adversarial patches generated by the model. This technique is similar 
to but different from training on defended traces, which are generated 
by the defender's model. In adversarial patch training, 
the classifier is iteratively trained on adversarial patches generated on 
the attacker model such that the classifier is robust to any type of 
adversarial patches. 

We evaluate \system{} against increasingly more robust classifiers on the
attacker side. 
Model robustness directly correlates with the number of epochs trained using 
adversarial training~\cite{li2020towards,wong2020fast,madry2017towards,tramer2019adversarial}. 
In our experiment, the defender uses the \dfsmall{} feature 
extractor (adversarially trained 
for $20$ epochs) to generate patches. We test the defense against attack classifier 
(Var-CNN) with varying robustness (adversarially trained from $20$ to $200$ 
epochs). Figure~\ref{tab:robust} shows the protection success rate 
remains $\leq 94\%$ against all the robust classifiers and does not trend downwards as
the model becomes more robust. This shows that 
generic adversarial training is less effective than training on 
defended traces, likely because the latter is more targeted towards 
specific types of perturbations generated by the defender.

\para{Training Orthogonal Classifiers.} Another countermeasure by the
attacker can be to explicitly avoid the features used by the defender to
generate patches, and to find other features for their trace
classification. If successful, it would produce a classifier that is largely
resistant to the patch. One approach is to build an attack classifier that
has orthogonal feature space as the defender's feature extractor.  The
attacker adds an additional loss term to model training to minimize the
neuron cosine similarity at an intermediate layer of the model. We train such
a classifier using \datasetlarge{} and Var-CNN model architecture.  In our
tests, the classifier only achieves $8\%$ normal classification accuracy
after $20$ epochs of training. This likely shows that there are not enough
alternative, orthogonal features that can accurately identify destination
websites from network traces. 

\para{Other Countermeasures Against Adversarial Patches.} There are other
defenses against adversarial patches explored in the computer vision
domain. However, most of them are limited to small input perturbations (less
than $5\%$ of the input)~\cite{akhtar2018defense,
  xiang2020patchguard,chiang2020certified}.  Others only work on contiguous
patches~\cite{naseer2019local, hayes2018visible,mccoyd2020minority}. To the
best of our knowledge, there exist no effective defense against larger
patches ($30\%$ of input) or nonconsecutive adversarial patches induced by
\system{}.  While it is always possible the community will develop more
effective defenses against larger adversarial patches, there exists a proven
lower bound on adversarial robustness that increases as the size of
perturbation increases~\cite{bhagoji2019lower}.  Thus, it is difficult to
be robust against large input perturbations without sacrificing
classification accuracy.



\section{Conclusion and Limitations}
\label{sec:discussion}

The primary contribution of \system{} is an effective defense against
website fingerprinting attacks (both traditional ML based and
DNN-based) that can run in real time to protect users. Our work is the
first to apply the concept of adversarial patches to WF defenses. 

However,
there are questions we have yet to study in detail. First, while most recent
defenses and attacks focus on a direction-only threat model and ignore
information leakage through time gaps~\cite{imani2019mockingbird,
  wang2017walkie,de2020trafficsliver,henri2020protecting,hou2020wf,sirinam2018deep},
some recent WF attacks~\cite{rahman2020tik,bhat2019var} also utilize time
gaps between packets to classify websites. We believe \system{} can be
extended to also defend against attacks utilizing time-gaps, and plan on
addressing this task in ongoing work. Second, we have not yet studied \system{}
deployed in the wild. Real measurements and tests in the wild may reveal additional
considerations,  leading to additional fine tuning of our system design.

{
 \small
 \bibliographystyle{acm}
 \balance
 \bibliography{wf_defense}
}

\clearpage

\appendix






\section{Effectiveness against Intersection Attacks}
\label{sec:trad}

Intersection attacks are popular attacks against anonymity 
systems~\cite{wright2004predecessor,
wright2008passive,mallesh2011analysis,danezis2003statistical}. 
In the context of website fingerprinting, an intersection
attacker assumes the victim visits the same websites
regularly and monitors the victim's network 
traffic for a longer time period (\eg every day 
over multiple days). Using the additional 
information, the attacker is able
to make better inferences on user's traces.
We test \system{} against an intersection attack 
intersection attack used by a previous
defense~\cite{imani2019mockingbird}. 

For each of the victim's browsing traces, the attacker 
logs the top-k results of the attack classifier (top-k is the $k$ 
output websites that have the highest probabilities that the trace belongs to). 
If a website consistently appear in the top-k results, 
then the attacker may believe that this site is in fact the website
that the user is visiting. We choose the same attack setup 
as \cite{imani2019mockingbird}. We assume attacker observes $5$ 
separate visits to the same website. 
Then the attacker saves the top-10 labels predicted 
by the classifier (using Var-CNN attack). 
The attack is successful if the selected label is the 
most frequently appeared label within joint list of $50$ labels 
($5$ sets of $10$ top-10 labels). We test on $20$ 
randomly selected websites from \datasetlarge. 
For all cases, the frequency of the correct website is 
far from the most frequent one (in the best case it ranks 
number $4$ out of $41$ websites) and the correct website never 
appears in all $5$ rounds. Thus, we conclude that intersection 
attacks is not effective against \system.

\begin{table*}
  \centering
  \resizebox{0.65\textwidth}{!}{
\begin{tabular}{ccccccc}
\hline
\multirow{2}{*}{\textbf{Dataset}} & \multirow{2}{*}{\textbf{\begin{tabular}[c]{@{}c@{}}Defender's\\ Feature Extractor\end{tabular}}} & \multicolumn{5}{c}{\textbf{Protection Success Rate Against WF Attacks}} \\ \cline{3-7} 
 &  & \textbf{k-NN} & \textbf{k-FP} & \textbf{CUMUL} & \textbf{DF} & \textbf{Var-CNN} \\ \hline
\multirow{4}{*}{\datasetsmall} & \dfsmall & 96\% & 97\% & 93\% & 95\% & 91\% \\
 & \dflarge & 96\% & 94\% & 97\% & 92\% & 93\% \\
 & \varsmall & 94\% & 92\% & 95\% & 93\% & 96\% \\
 & \varlarge & 97\% & 96\% & 95\% & 95\% & 94\% \\ \hline
\multirow{4}{*}{\datasetlarge} & \dfsmall & 94\% & 92\% & 95\% & 93\% & 92\% \\
 & \dflarge & 95\% & 94\% & 96\% & 97\% & 91\% \\
 & \varsmall & 94\% & 97\% & 98\% & 94\% & 92\% \\
 & \varlarge & 96\% & 95\% & 96\% & 95\% & 97\% \\ \hline
\end{tabular}
  }
  \caption{Protection performance of \system{} using non-robust feature extractor against different WF attacks when transfering to classifiers trained on different datasets and/or architecture. 
  }
  \label{tab:transfer_nonrobust}
\end{table*}

\begin{figure*}
\begin{minipage}{0.32\textwidth}
  \centering
  \includegraphics[width=1\textwidth]{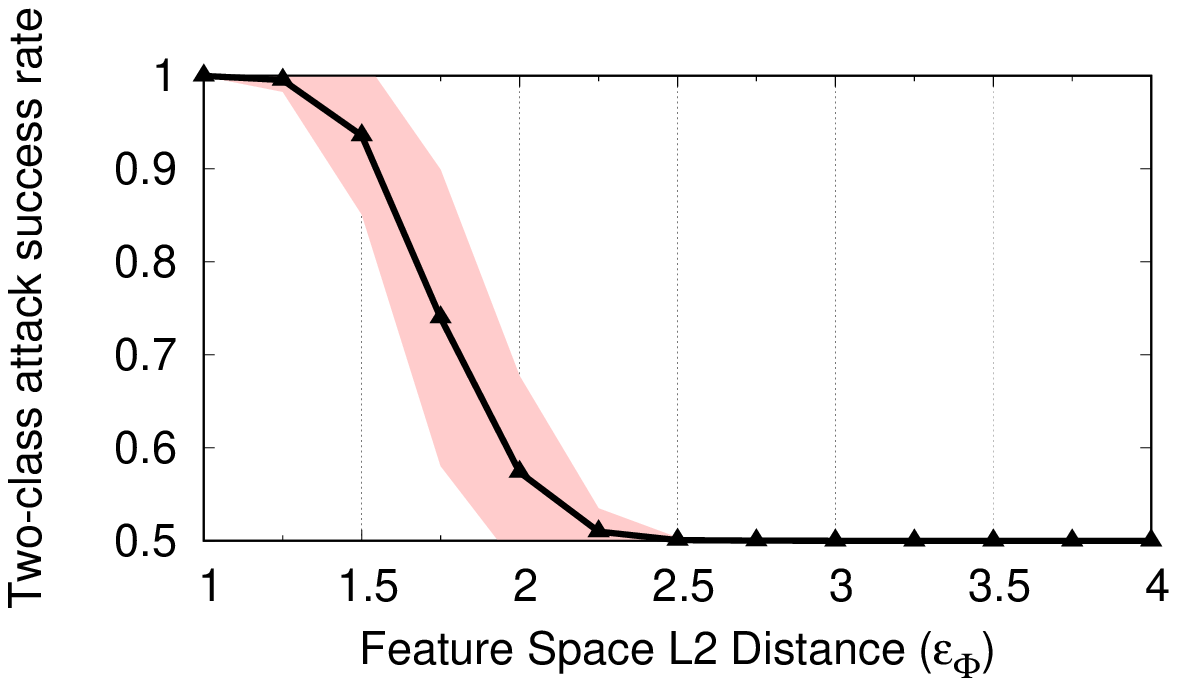}
  \caption{Upper bound on the effectiveness of \emph{any attack classifier} for a fixed feature extractor $\Phi$, averaged over $500$ choices of source-pairs.}
  \label{fig:theory1}
  \end{minipage}
  \hfill
  \begin{minipage}{0.32\textwidth}
    \centering
  \includegraphics[width=1\textwidth]{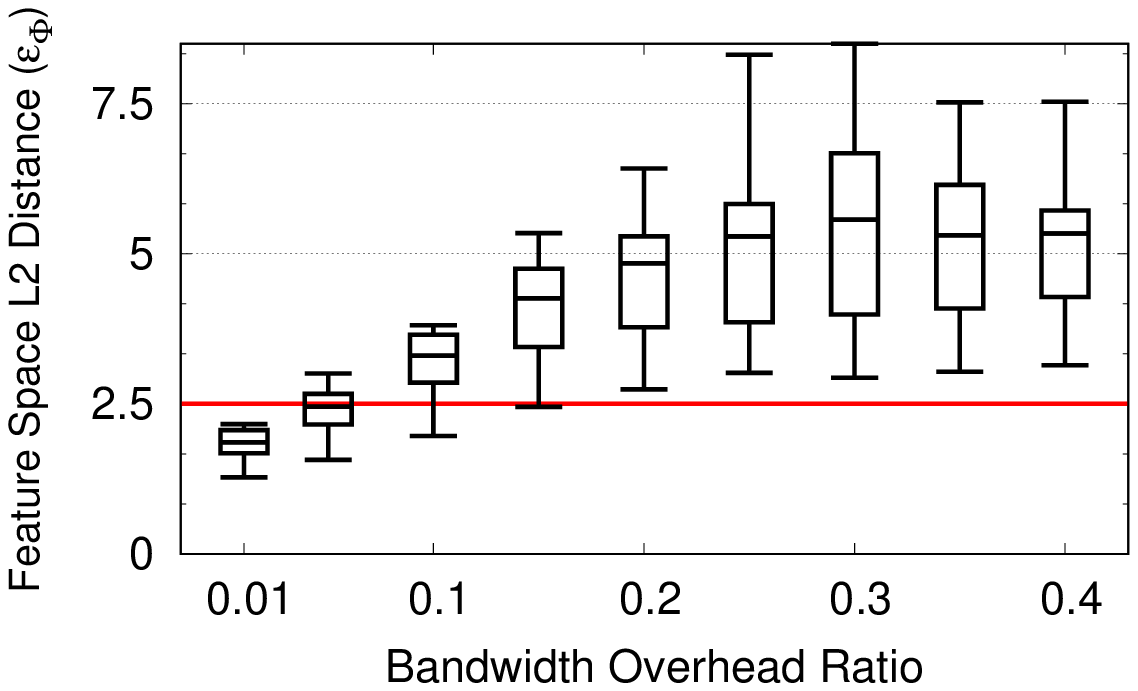}
\caption{Variation in the distance moved in feature space with the input bandwidth overhead, averaged over a 100 different targets for a fixed source.}
  \label{fig:theory2}
  \end{minipage}
  \hfill
\begin{minipage}{0.32\textwidth}
	\centering
\includegraphics[width=1\textwidth]{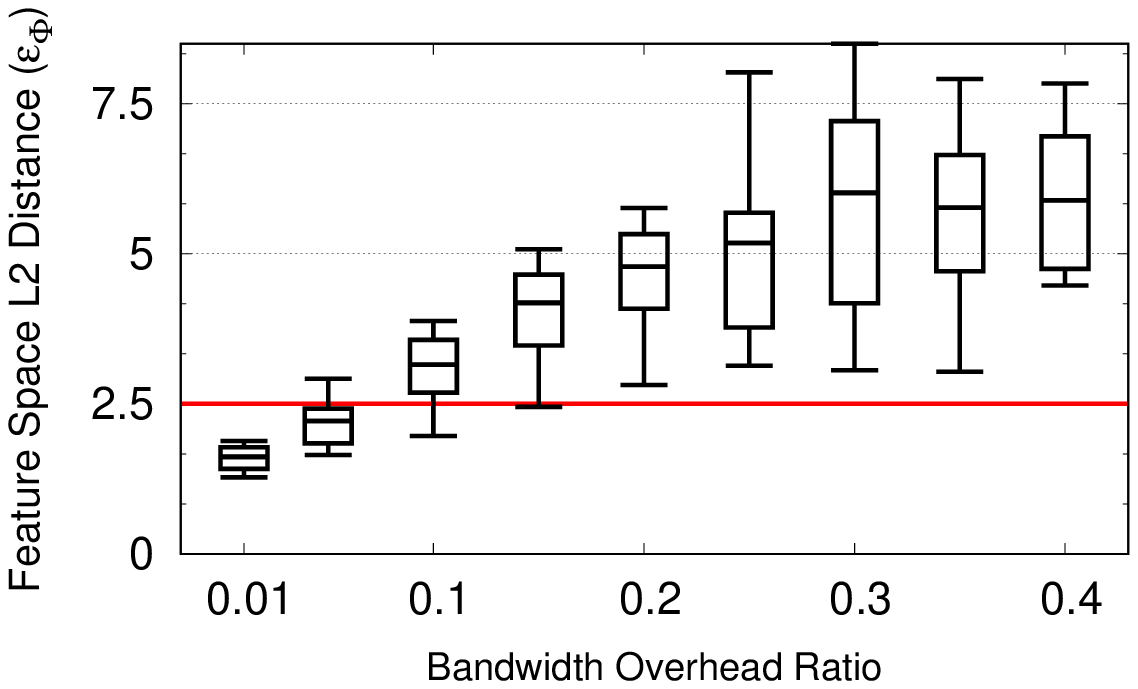}
\caption{Variation in the distance moved in feature space with the input bandwidth variation, averaged over a single target for 20 different sources.}
\label{fig:theory3}
\end{minipage}
\end{figure*}

\section{Theoretical Justification of Defense}
\label{sec:advance}

We show that \system{} provides provable robustness guarantees when
both the attacker and defender use the same fixed feature extractor
$\Phi$.

We use recent theoretical results~\cite{bhagoji2019lower,pmlr-v119-pydi20a} on learning $2$-class classifiers in the presence of adversarial examples which has shown that as the strength of the attacker (defender in our case) increases, the $0-1$ loss incurred by \emph{any classifier} is lower bounded by a transportation cost between the conditional distributions of the two classes. In the case of image classification, which is the example considered in previous work, the budget is typically too small to observe interesting behavior in terms of this lower bound. 

However, since we consider network traffic traces to which much larger amounts of perturbation can be added, we encounter non-trivial regimes of this bound. This implies that with a sufficiently large bandwidth, no classifier used by the attacker will be able to distinguish between traces from the source and target classes. In order to demonstrate this, we make the following assumptions:
\begin{enumerate}
  \item The attacker is attempting to distinguish if a trace $x$ belongs to the original class $W$ or the target class $T$, with distributions $P_W$ and $P_T$ respectively, both of which may be defended (i.e. adversarially perturbed).
  \item The attacker uses a classifier function $F$ acting over the feature space $\Phi(\mathcal{X})$, where $F$ can be any measurable function and $\Phi: \mathcal{X}\rightarrow\mathbb{R}^k$ is any fixed feature extractor. The resulting end-to-end classifer is represented by the tuple $(\Phi,F)$.
  \item The defender has an $\ell_2$ norm perturbation budget of $\epsilon_{\Phi}$ in the feature space, matching the choice of $\text{Dist}(\cdot,\cdot)$ in Eq.~\ref{eq:adv_patch}. The feature space budget is related to the input space bandwidth overhead $R$ by a mapping function $M$ that maps balls of radius $R$ in the input space to balls of radius at least $\epsilon_{\Phi}$ in the feature space.
\end{enumerate}

Given these assumptions, we can now state the following theorem, adapted from Bhagoji et al. \cite{bhagoji2019lower}:

\begin{theorem}[Upper bound on attacker success]
For any pair of classes $W$ and $T$ with conditional probability distributions $P_W$ and $P_T$, and the joint distribution $P$, over $\mathcal{X}$, and with a fixed feature extractor $\Phi$, the attack success rate of any classifier $F$ is 
\begin{align}
  ASR((\Phi,F),P,R) \leq \frac{1}{2} \left( 1 + C_{\epsilon_{\Phi}}(\Phi(P_W),\Phi(P_T)) \right).
\end{align}
\label{thm: upper_bound}
\end{theorem}

\begin{proof}
The end-to-end classifier has a fixed feature extractor $\Phi$ and a classifier function $F$ that can be optimized over. To determine an upper bound on the attack success rate achievable by a classifier F, we have
\begin{align*}
	ASR((F,\Phi),P,R) &=\underset{x \sim P}{\mathbb{E}}  \left[ \min_{|\tilde{x}-x|\leq R} \bm{1} \left((F,\Phi)(  \tilde{x})=y\right)  \right],\\
	&\leq \underset{\Phi(x) \sim \phi(P)}{\mathbb{E}}  \left[\min_{ ||\Phi(\tilde{x}) - \Phi(x)||_2 \leq \epsilon_{\Phi}}  \bm{1} \left(F(\Phi(\tilde{x}))=y\right)  \right],\\
	&=ASR(F,\Phi(P),\epsilon_{\Phi})
\end{align*}
The $\leq$ arises from a conservative estimate of the distance moved in feature space.
Having transformed the attack success rate calculation to one over the feature space, we can now directly apply Theorem 1 from \cite{bhagoji2019lower}, which gives
\begin{align*}
	\max_F \, ASR(F,\Phi(P),\epsilon_{\Phi}) = \frac{1}{2} \left( 1 + C_{\epsilon_{\Phi}}(\Phi(P_W),\Phi(P_T)) \right).
\end{align*}
From~\cite{bhagoji2019lower},
\begin{multline}
	C(\Phi(P_W),\Phi(P_T))\\= \underset{P_{WT} \in \Pi(P_W,P_T) }{\inf} \underset{(x_W,x_T) \sim P_{WT}}{\mathbb{E}} \left[\bm{1}(|\Phi(x_W)-\Phi(x_T)| \geq 2 \epsilon_{\Phi} )\right],
 \end{multline}

where $\Pi(P_W,P_T)$ is the set of joint distributions over $\mathcal{X}_W \times \mathcal{X}_T$ with marginals $P_W$ and $P_T$.
\end{proof}

The main takeaway from the above theorem is that the better separated the perturbed feature vectors from the two classes are, the higher the attack success rate will be. Thus, from the defender's perspective, the bandwidth $R$ has to be sufficient to ensure that the resulting $\epsilon_{\Phi}$ leads to low separability.

\para{Empirical upper bounds.} We compute the feature space distances between 500 different sets of source $W$ and target $T$ websites and plot the maximum attack success rate as the budget $\epsilon_{\Phi}$ in the feature space is varied (Figure~\ref{fig:theory1}). We use a robust feature extractor trained on the \datasetlarge dataset to derive this upper bound. With a feature space budget of $2.5$, the attack success rate drops to $50\%$, which for a two-class classification problem implies that no classifier can distinguish between the two classes.

Now, it remains to be established that a reasonable input bandwidth overhead can lead to a feature space budget of $2.5$. In Figures~\ref{fig:theory2} and~\ref{fig:theory3}, we see that varying over 20 choices of sources with respect to a fixed target and 100 targets for a fixed source, the minimum distance moved in the feature space is larger than $2.5$, making it a conservative estimate for $\epsilon_{\Phi}$ in Theorem~\ref{thm: upper_bound}.

\para{Remarks} We note that our analysis here is restricted to the $2$-class setting, thus the conclusions drawn may not apply for all source-target pairs. We hypothesize that this explains the lower than $100\%$ protection rate at $R=0.3$  we observe in \S\ref{sec:eval}, since some pairs of classes may not be sufficiently well-separated in feature space. We also note that the analysis used here improves over that of Cherubin \cite{cherubin2017bayes}, where the effectiveness of an attack is only tested against a fixed defense, while we provide optimal bounds in the setting where the attacker and defender can adapt to one another.

\end{document}